\documentclass[12pt]{article}
\usepackage{verbatim}
\usepackage{caption}
\usepackage{adjustbox}
\usepackage{amsfonts}
\usepackage{graphics}
\usepackage{amsmath}
\usepackage{times}
\usepackage{appendix}
\usepackage{color}
\usepackage{mathtools}
\usepackage{enumerate}
\usepackage{fancyhdr,latexsym,amsmath,amsfonts,amssymb,amsbsy,amsthm,url}
\usepackage[margin=0.5in,footskip=0.25in]{geometry}
\usepackage{graphics,graphicx,epsfig}
\usepackage{breqn}
\usepackage{caption,subcaption,float,subfloat}
\usepackage{pdflscape}
\usepackage{hyperref}
\usepackage[ruled,vlined]{algorithm2e}
\newtheorem{theorem}{Theorem}[]
\newtheorem{lemma}{Lemma}[]

\newtheorem{proposition}{Proposition}[]
\DeclarePairedDelimiter\ceil{\lceil}{\rceil}
\DeclarePairedDelimiter\floor{\lfloor}{\rfloor}
\usepackage[english]{babel}
\usepackage[dvipsnames]{xcolor}
\makeatletter
\providecommand{\keywords}[1]
{
  \small	
  \textit{Keywords:} #1
}

\renewcommand{\section}{
	\@startsection
	{section}% name
	{1}% level
	{0pt}% indent
	{1.1\baselineskip}% beforeskip
	{0.2\baselineskip}% afterskip
	{\sc \centering}% style
}

\renewcommand{\subsection}{
	\@startsection
	{subsection}% name
	{1}% level
	{0pt}% indent
	{1.1\baselineskip}% beforeskip
	{0.2\baselineskip}% afterskip
	{\sc \centering}% style
}

\renewcommand{\subsubsection}{
	\@startsection
	{subsubsection}% name
	{1}% level
	{0pt}% indent
	{1.1\baselineskip}% beforeskip
	{0.2\baselineskip}% afterskip
	{\sc \centering}% style
}

\makeatother

\newcommand{\innerproduct}[2]{\langle #1, #2 \rangle}
\newcommand{\norm}[1]{\lVert #1 \rVert}
\newcommand{\bnorm}[1]{\bigg \lVert #1 \bigg \rVert}
\newcommand{\abs}[1]{\lvert #1 \rvert}
\newcommand{\babs}[1]{\bigg \lvert #1 \bigg \rvert}
\usepackage[flushleft]{threeparttable}
\usepackage{rotating,booktabs,multirow}
\usepackage{colortbl}
\usepackage{makecell,cellspace,caption}

\begin{document}
	
\title{\large\sc Multilevel Richardson-Romberg and Importance Sampling in Derivative Pricing}
\normalsize
\author{\sc{Devang Sinha} \thanks{Department of Mathematics, Indian Institute of Technology Guwahati, Guwahati-781039, Assam, India, e-mail: dsinha@iitg.ac.in}
\and \sc{Siddhartha P. Chakrabarty} \thanks{Department of Mathematics, Indian Institute of Technology Guwahati, Guwahati-781039, Assam, India, e-mail: pratim@iitg.ac.in,
Phone: +91-361-2582606, Fax: +91-361-2582649}}
\date{}
\maketitle

\begin{abstract}
In this paper, we propose and analyze a novel combination of multilevel Richardson-Romberg (ML2R) and importance sampling algorithm, with the aim of reducing the overall computational time, while achieving desired root-mean-squared error while pricing. We develop an idea to construct the Monte-Carlo estimator that deals with the parametric change of measure. We rely on the Robbins-Monro algorithm with projection, in order to approximate optimal change of measure parameter, for various levels of resolution in our multilevel algorithm. Furthermore, we propose incorporating discretization schemes with higher-order strong convergence, in order to simulate the underlying stochastic differential equations (SDEs) thereby achieving better accuracy. In order to do so, we study the Central Limit Theorem for the general multilevel algorithm. Further, we study the asymptotic behavior of our estimator, thereby proving the Strong Law of Large Numbers. Finally, we present numerical results to substantiate the efficacy of our developed algorithm. 
\end{abstract}


\section{Introduction}
\label{Section_One_Introduction}

At the core of financial engineering lie the fundamental questions of pricing, portfolio management, and risk mitigation, which for all practical purposes, necessitates resorting to computational techniques, particularly, the Monte Carlo simulation approach. The emergence of this technique, as the key computational approach in the finance industry, can be attributed to the requirement of simulating high-dimensional stochastic models, which in turn is a result of the growth in computational complexity, pertaining to the dimension of the problem itself. More specifically, the main purpose of the usage of this procedure is to numerically approximate the expected value $E(Y)$, where $Y=P(X)$ is a functional of the random variable $X$, with (in case of the concerned applications) $X$ being driven by a stochastic differential equation (SDE).

The increasing interest and application of the Monte Carlo approach in finance industry, can also be attributed to the development of Multilevel Monte Carlo (MLMC) technique, due to Giles \cite{giles2008multilevel}, as it delivers remarkable improvement in computational complexity, over the standard Monte Carlo technique, in the biased framework. The interested readers may refer to the website of Giles ( \url{http://people.maths.ox.ac.uk/~gilesm/mlmc_community.html}) for an exhaustive and detailed listing of the flow of research concerning MLMC, both theoretical as well as application. The technique of MLMC draws its foundations from the technique of parametric integration, developed in \cite{heinrich2001multilevel}, which was extended in \cite{giles2008multilevel}, in order to construct a MLMC path method. While the control variate viewpoint used in \cite{giles2008multilevel} is similar to that of \cite{heinrich2001multilevel}, however, unlike in case of \cite{heinrich2001multilevel}, for the construction in \cite{giles2008multilevel}, the random variable is infinite dimensional, with respect to the Brownian path, and there is no parametric integration that is involved. MLMC, as the name suggests, makes use of various levels of resolution, from the coarsest to the finest, where, from the perspective of the control variate, the simulation executed using the coarsest path, is used as a control variate, while carrying out the finer path based estimation. The introduction to MLMC was accomplished through the Euler-Maruyama discretization, a thorough analysis of which suggested that in order to estimate the variance of MLMC, one has to examine the strong convergence properties.

In another 2008 article \cite{giles2008improved}, Giles presented a Milstein scheme for discretization of an SDE, an approach that led to improvement in the estimation of the variance. Further, it was demonstrated that, it could be more prudent to make use of different estimators for the coarser level and the finer level. The numerical results presented, as well as the thorough numerical analysis carried out in \cite{giles2019analysis}, demonstrates the efficacy of the method in case of option pricing problems. With the improvement in variance convergence, resulting from the usage of the Milstein scheme, the extension of the scheme in case of higher dimensional SDEs (involving the concept of Levy area) can be accomplished. However, in practice, there is no efficient way to simulate the Levy area, and the interested reader may refer to \cite {giles2013multilevel} for more discussion on this. In order to tackle this issue, the authors in \cite {giles2013antithetic}, introduced the antithetic MLMC estimator, based on the classical antithetic variance reduction technique, wherein the idea of antithetic MLMC exploits the flexibility of the general MLMC estimator.

Giles and Waterhouse \cite{giles2009multilevel}, published an extension of MLMC using Quasi-Monte Carlo (QMC) sample paths instead of Monte Carlo sample paths. The numerical results presented showed the effectiveness of QMC for SDE applications. However, the results presented in the paper were not supported by any theoretical development in this context. Nevertheless, recently, there has been some promising contribution towards the theoretical development of Quasi MLMC (QMLMC), for which one may refer to \cite{giles2015multilevel} and references therein. It has been shown that, under certain conditions, the QMLMC leads to multilevel methods, with a complexity which is $O(\epsilon^{-p})$ with $p<2$. In \cite{rhee2015unbiased}, Rhee and Glynn introduced a new approach of constructing an unbiased estimator, given a family of the biased estimators. The idea presented by them is closely related to that of MLMC, wherein the finest level of estimation is chosen, contingent on the level of accuracy.  The results presented in \cite{rhee2015unbiased} demonstrated a significant improvement in the computational cost over the standard MLMC. Also, they proved the square root convergence of the estimator, given that the strong order of convergence is greater than $\displaystyle{\frac{1}{2}}$ for the path functionals.

\section{Preliminaries}

In this section we discuss the necessary preliminaries and assumptions, as a prelude to the comprehensive study of the algorithm developed in this work. We present a brief outline of Multilevel Richardson-Romberg estimator, introduced in \cite{lemaire2017multilevel} and recall the importance sampling algorithm in the context of the standard Monte-Carlo algorithm.

\subsection{Multilevel Richardson-Romberg}
\label{section_one_ml2r}

In \cite{giles2008multilevel}, Giles explored the Richardson extrapolation in the context of both the MLMC and the standard Monte Carlo. The MLMC on its own was significantly better than the Richardson extrapolation. However, taken together, they worked even better. Lemaire and Pages  \cite{lemaire2017multilevel} took this approach and undertook further comprehensive error analysis. They combined the method developed in \cite{giles2008multilevel} and Multistep Richardson extrapolation in order to minimize the cost of simulation.

Suppose, we are interested in estimating the expected payoff value \textit{i.e.,} $\mathbb{E}\left(P(X_T)\right)$ in an option pricing problem, with finite time horizon $T>0$, where $(X_{t})_{0\leq t \leq T}$ is a process with values in $\mathbb{R}^{d}$, governed by the following SDE,
\begin{equation}
\label{oneeq:1}
dX_{t}=b(X_{t})dt+\sum\limits_{j=1}^{q} \sigma_{j}(X_{t})dW_{t}^{j},~X_{0}=x \in \mathbb{R}^{d}, 
\end{equation}
where, $W:=\begin{pmatrix}W_{1}&W_{2}&\dots&W_{q}\end{pmatrix}$ is a $q$-dimensional Brownian motion on a filtered probability space\\ $\left(\Omega,(\mathcal{F}_{t}\right)_{0 \leq t \leq T},\mathbb{P})$, with $b:\mathbb{R}^{d} \rightarrow \mathbb{R}^{d}$ and $\sigma_{j}:\mathbb{R}^{d} \rightarrow \mathbb{R}^{d}$ being the functions satisfying the following condition:
\begin{equation}
\label{onecond:A_1}
\forall x,y \in \mathbb{R}^{d},~\lvert b(x)-b(y)\rvert+\sum\limits_{j=1}^{q}\lvert \sigma_{j}(x)-\sigma_{j}(y)\rvert < K_{b,\sigma}\lvert x-y \rvert,~\text{where}~K_{b,\sigma}>0. 
\tag{A.1}
\end{equation}
The assumption (\ref{onecond:A_1}) ensures the existence and uniqueness \cite{kloeden1992stochastic} of the strong solution of \eqref{oneeq:1}. Now, in order to estimate $\mathbb{E}\left(P(X_T)\right)$ one should be able to simulate $(X_{t})_{0\leq t \leq T}$. However, except for a handful of cases, where one can devise an analytical or semi-analytical solution, we must resort to discretization schemes to perform the simulation. As discussed above, we use the Euler-Maruyama discretization scheme in order to simulate equation \eqref{oneeq:1}.

Let, $M \in \mathbb{N}$ be the refinement factor and let $n_{l}=M^{l-1}$, where $l=1,\dots,L$ is the level refinement. Let $\displaystyle{h_{l}=\frac{T}{n_{l}}}$ be the time step size on level $l$. For example, the Euler scheme on level $l$ is given by,
\[dX_{t}^{n_{l}}=b\left(X^{n_{l}}{\eta_{n_{l}(t)}}\right)dt+\sum\limits_{j=1}^{q}\sigma_{j}\left(X^{n_{l}}{\eta_{n_{l}(t)}}\right)dW_{t}^{j},~l=1,\dots,L,\]
where, $\eta_{n_{l}(t)}:=\floor{t/h_l}h_{l}$. We approximate $\mathbb{E}\left[P\left(X_{T}\right)\right]$ by 
$\mathbb{E}\left[P\left(X_{T}^{n_{L}}\right)\right]$, where $L$ is the finest level. In \cite{lemaire2017multilevel} the authors proposed a methodology combining the order bias cancellation of Multilevel Richardson Romberg (ML2R), with variance control of MLMC, to solve the problem of improving the computational complexity, along with determining the optimal parameters in order to achieve the desired root-mean-squared error, with minimum computational effort. The study was based on \textit{weak error} and \textit{strong error} assumptions, on the sequence $\left(P(X_{T}^{h})\right)_{h \in \mathcal{H}}$, where $\mathcal{H}=\left\{\frac{\mathbf{h}}{n}, n > 1\right\}$. The assumptions \cite{lemaire2017multilevel,kloeden1992stochastic,giorgi2017limit} are stated as follows:
\begin{enumerate}[(A)]
\item \textit{Weak Error}:
\begin{equation}
\label{oneeq:2}
\exists \alpha > 0, \Bar{L} \geq 1, (c_{l})_{1\leq l \leq \Bar{L}}, \quad \mathbb{E}[P(X^{h}_{T})]-\mathbb{E}[P(X^{0}_{T})]=\sum\limits_{k=1}^{\Bar{L}} c_{k}h^{\alpha k}+h^{\alpha \Bar{L}}\eta_{\Bar{L}}(h).  
\end{equation}
\item \textit{Strong Error}:
\begin{equation}
\label{oneeq:3}
\exists \beta>0, V_{1} \geq 1, \|P(X^{h}_{T})-P(X^{0}_{T})\|^{2}_{2}=\mathbb{E}\left[|P(X^{h}_{T})-P(X^{0}_{T})|^{2}\right] \leq V_{1}h^{\beta}.
\end{equation}
\end{enumerate}
With the consideration of the above assumptions, the ML2R estimator is defined as,
\begin{equation}
\label{oneeq:4}
\mathbf{J}^{N}_{\pi} \coloneqq \mathbb{E}\left[P\left(X_{T}^{n_{L}}\right)\right]= \frac{1}{N_{1}}\sum\limits_{k=1}^{N_{1}}P\left(X_{T}^{n_{1},k}\right)+\sum\limits_{l=2}^{L}\frac{\widetilde{W}_{l}}{N_{l}}\sum\limits_{k=1}^{N_{l}}\left(P\left(X_{T}^{n_{l},k}\right)-P\left(X_{T}^{n_{l-1},k}\right)\right),
\end{equation}
where $\pi=(h,\mu,L)$ are the optimal parameters obtained as the solution to,
\begin{equation}
\label{oneeq:5}
(\pi(\epsilon),N(\epsilon))={\arg\min}_{\|\mathbf{J}^{N}_{\pi}-\mathbf{J}_{0}\|_{2} \leq \epsilon}\text{Cost}(\mathbf{J}^{N}_{\pi}).
\end{equation}
where the cost function is given by \cite{giorgi2017limit},
\[\text{Cost}(\mathbf{J}^{N}_{\pi})=\frac{N}{h}\sum\limits_{l=1}^L\mu_{l}(n_{l-1}+n_{l}).
\]
Further, in the above equation, $N_{l}$ is the number of sample paths on level $l$, and $\widetilde{W}_{l}$ are the weights given by,
\begin{equation}
\label{oneeq:6}
\widetilde{W}_{l}=\sum\limits_{j=l}^{L}w_{l},~l=1,\dots,L,
\end{equation}
where $\textbf{w}=\left(w_{l}\right)_{1 \leq l \leq L}$ is the solution to the Vandermonde system $V \textbf{w}=e_{1}$, with the Vandermonde matrix being defined by,
\begin{equation}
\label{oneeq:7}
V=\begin{pmatrix}
1 & 1 & \dots & 1\\
1 & n_{2}^{-\alpha} & \dots & n_{L}^{-\alpha}\\
\vdots & \vdots & \dots & \vdots\\
1 & n_{2}^{-\alpha(L-1)} & \dots & n_{L}^{-\alpha(L-1)}\\
\end{pmatrix}
\end{equation}
Here, $\alpha$ is the weak error rate as defined above. The interested reader may refer to \cite{lemaire2017multilevel} for the construction of the optimal parameters, as the closed solution to equation \eqref{oneeq:5}. Here we tabulate the explicit values of these parameters required to achieve the root-mean-squared error $\epsilon$, with the following constants being used:
\begin{enumerate}[(A)]
\item $\displaystyle{\lambda =\sqrt{\frac{V_{1}}{\text{Var}\left(P(X^{0}_{T})\right)}}~\text{and}~\widetilde{c}_{\infty}=\lim_{L \rightarrow \infty}|c_{L}|^{1/\alpha} \in (0,\infty)}$.
\item $\displaystyle{\underline{C}_{M,\beta}=\frac{1+M^{\beta/2}}{\sqrt{1+M^{-1}}}~\text{and}~\overline{C}_{M,\beta}=\left(1+M^{\beta/2}\right)\sqrt{1+M^{-1}}}$.
\end{enumerate}

\begin{table}[H]
\begin{center}
\begin{tabular}{|c|c|}
\hline            
$L(\epsilon)$ & \\
&$\left\lceil \frac{1}{2}+\frac{\log(\widetilde{c}_{\infty}^{1/\alpha}\mathbf{h})}{\log(M)}+\sqrt{\left(\frac{1}{2}+\frac{\log(\widetilde{c}_{\infty}^{1/\alpha}\mathbf{h})}{\log(M)}\right)^{2}+ 2\frac{\log(A/\epsilon)}{\alpha\log(M)}}\right\rceil$\\ 
& \\\hline
$h(\epsilon)$ & $\mathbf{h}$ \\ \hline
$\mu(\epsilon)$ &\\
&$\mu_{1}=q^{*}(1+\lambda h^{\frac{\beta}{2}})$\\
&$\mu_{l}=q^{*}\lambda h^{\frac{\beta}{2}} \underline{C}_{M,\beta}|\widetilde{W}_l(L,M)|M^{-\frac{(1-\beta)}{2}(l-1)}, l=2,\dots ,L$,~$\sum\limits_{1\leq l \leq L}\mu_{l}=1$\\ 
& \\ \hline
$N(\epsilon)$ & \\
&$\left(1+\frac{1}{2 \alpha L}\right)\frac{\text{Var}\left(P(X^{0}_{T})\right)\left(1+ \lambda h^{\beta/2}+\lambda h^{\beta/2}\overline{C}_{M,\beta}\sum\limits_{l=2}^{L}|\widetilde{W}_l(L,M)|M^{\frac{1-\beta}{2}(j-1)} \right)}{\epsilon^{2}q^{*}}$ \\ 
&\\\hline
\end{tabular}
\end{center}
\caption{Optimal parameters for the ML2R estimator \label{T:Optimal Parameters}}
\end{table}
The above estimator was highly effective whenever there is a strong order of convergence \textit{i.e.,} $\beta \leq 1$, as it achieves $O\left(\epsilon^{-2}\log(1/\epsilon)\right)$ for $\beta =1$ and $O\left(\epsilon^{-2}e^{\frac{1-\beta}{\sqrt{\alpha}}\sqrt{2\log(1/\epsilon)\log(M)}}\right)$ for $\beta < 1$, contrary to $O\left(\epsilon^{-2}\log(1/\epsilon)^2\right)$ and $O\left(\epsilon^{-2-\frac{1-\beta}{\alpha}}\right)$, respectively, which is achieved by the standard MLMC.
\subsection{Importance Sampling}
Now following the idea of \cite{arouna2004adaptative}, we consider a parametric family of stochastic process $\left(X_{t}(\theta)\right)_{0 \leq t \leq T}$, with $\theta \in \mathbb{R}^{d}$, governed by the following SDE,
\begin{equation}
\label{oneeq:8}
dX_{t}(\theta)=(b(X_{t}(\theta))+\sigma(X_{t}(\theta))\theta)dt+\sum\limits_{j=1}^{q}\sigma_{j}(X_{t}(\theta))dW_{t}^{j},~\sigma(x)=\begin{pmatrix}\sigma_{1}(x)&\dots&\sigma_{q}(x)\end{pmatrix}.
\end{equation}
By the Girsanov's Theorem, we know that there exists a probability measure $\mathbb{P}_{\theta}$ equivalent to $\mathbb{P}$ such that,
\begin{equation}
\label{oneeq:9}
\frac{d\mathbb{P}_{\theta}}{{d\mathbb{P}|}_{\mathcal{F}_t}}=\exp{\left(-\innerproduct{\theta}{ W_{t}}-\frac{1}{2}\lVert\theta\rVert^{2}t\right)} \triangleq \mathcal{J}^-(W_t,\theta),
\end{equation}
under which, the process $\left(\theta t+W_{t}\right)_{0\leq t \leq T}$ is a Brownian motion. Therefore,
\begin{equation}
\label{oneeq:10}
\mathbb{E}_{\mathbb{P}}\left[P(X_{T})\right]=\mathbb{E}_{\mathbb{P}_{\theta}}\left[P(X_{T}(\theta))\right]=\mathbb{E}_{\mathbb{P}}\left[P(X_{T}(\theta))\mathcal{J}^-(W_T,\theta)\right].
\end{equation}

The efficiency of the Monte Carlo method is considerably dependent on how small the variance is in the process of estimation. Among the many variance reduction techniques, available in literature (for improvement of efficiency of the Monte Carlo method), importance sampling is one such variance reduction technique, which is known for its efficiency \cite{glasserman2004monte}. In general, the idea of parametric importance sampling is to introduce the parametric transformation such that $\forall \theta \in \mathbb{R}^{d}$,
\[\mathbb{E}\left[P(X_{T})\right]=\mathbb{E}\left[h(\theta, X_{T})\right].\]
More specifically, we consider the general problem of estimating $\mathbb{E}[P(X)]$, where $X$ is a $d$-dimensional random variable. Also if $f(x)$ is the multivariate density, then,
\[\mathbb{E}[P(X)]=\int P(x)f(x)dx=\int P(x+\theta)f(x+\theta)dx=\int h(\theta, x)f(x)dx,\] where $\displaystyle{h(\theta,x)=\frac{P(x+\theta)f(x+\theta)}{f(x)}}$.
Now the idea of importance sampling Monte Carlo method is to estimate $\mathbb{E}\left(P(X_{T})\right)$, where $\theta$ is given by,
\begin{equation}
\label{oneeq:11}
\theta^{*}=\arg\min_{\theta \in \mathbb{R}^{d}} \text{Var}~\left(P(X_{T}(\theta))\mathcal{J}^-{(W_T,\theta)}\right).  
\end{equation}

In order to solve the above problem one can use the so-called Robbins-Monro algorithm that deals with a sequence of random variable $\left(\theta_{i}\right)_{i\in\mathbb{N}}$, which approximates $\theta^{*}$ accurately. However, the convergence of this algorithm requires restrictive conditions known as the non explosion condition given by \cite{alaya2015importance},
\[\mathbb{E}[h^{2}\left(\theta, X_{T}\right)] \leq C\left(1+|\theta|^{2}\right)~\text{for all}~\theta \in \mathbb{R}^{d}.\] 
To overcome these restrictive conditions, a truncation based procedure was introduced by Chen in \cite{chen1987convergence,chen1986stochastic} which was furthered in \cite{andrieu2005stability,lelong2008almost}. This was then proposed in the context of importance sampling \cite{arouna2004adaptative}. An unconstrained procedure to approximate $\theta^{*}$, by using the regularity of the involved density extensively, was introduced in \cite{lemaire2010unconstrained} along with the proof of the convergence of the algorithm.

The authors in \cite{alaya2015importance} studied importance sampling in the context of statistical Romberg integration, which is a one-step generalization of the standard Monte Carlo integration. The paper used both constrained, as well as unconstrained optimization routines to approximate $\theta^{*}$. Further, they proved the almost sure convergence of both constrained and unconstrained versions of the optimization algorithm, towards the optimal parameter $\theta^{*}$, under discretized diffusion setting. The idea was then extended by \cite{alaya2016improved} and \cite{kebaier2018coupling} in the context of MLMC. The study carried out in \cite{alaya2016improved} is the extension of \cite{alaya2015importance} in the paradigm of MLMC. However \cite{kebaier2018coupling} used sample average approximation to approximate $\theta^{*}$. Also, their algorithm performed an approximation procedure on each level of resolution, thus further enhancing the overall performance of the adaptive algorithm.

In this paper, we develop a hybrid algorithm incorporating ML2R and adaptive importance sampling procedure developed in \cite{arouna2004adaptative} in the context of standard Monte Carlo in order to reduce the overall computational time, while achieving the desired accuracy. In \cite{alaya2015importance}, the authors studied the discretized version of the asymptotic variance obtained in the central limit theorem associated with the statistical Romberg method. The study was performed to determine the optimal parameter $\theta^{*}$, to minimize the variance of the Monte Carlo estimator. In \cite{kebaier2018coupling}, the authors extended the algorithm developed in \cite{alaya2015importance} in the context of MLMC, building upon the Central Limit Theorem for the Euler MLMC studied in \cite{alaya2015central}. As discussed in \cite{alaya2015central}, the optimal value of $\theta$ is obtained by using a constrained and unconstrained version of the stochastic Robbins-Monro algorithm, whereas in \cite{kebaier2018coupling} it is determined using deterministic Newton-Raphson's method. However, the algorithm was developed in the context of Euler discretization, and only European option pricing was studied, with no reference to path-dependent functionals. \textbf{In this paper, we develop upon the Central Limit Theorem, for the general multilevel approach studied in \cite{giorgi2017limit}, that incorporates the results for path-dependent functionals as well. We look at the stochastic optimization approach to approximate $\theta^*$. Also, we use the Milstein scheme, as well as Euler in order to discretize the underlying SDE.} In order to highlight the novelty and contribution of this work, we present a comparison of the work accomplished, with the existing literature, in Table     \ref{tab:comparison}

\begin{table}[!h]
\centering
\begin{adjustbox}{width=\columnwidth,center}
\begin{tabular}{cccccccc}
\textbf{ML2R} & \textbf{IS} & \textbf{Hybrid Algorithm: (IS+MC Version)} & \textbf{Euler Scheme} & \textbf{Milstein Scheme} & \textbf{Vanilla Option} & \textbf{Exotic Option} & \textbf{Reference} \\
Yes & No  & No                                         & Yes & No  & Yes & Yes & \cite{lemaire2017multilevel} \\
No  & Yes & Yes ( With Standard Monte carlo)           & Yes & No  & Yes & No  & \cite{arouna2004adaptative} \\
No  & Yes & Yes (With MLMC)                            & Yes & No  & Yes & No  & \cite{kebaier2018coupling} \\
No  & Yes & Yes (With Statistical Romberg Integration) & Yes & No  & Yes & No  & \cite{alaya2015importance}     \\
No  & Yes & Yes (Adaptive MLMC)                        & Yes & No  & Yes & No  & \cite{alaya2016improved} \\
Yes & Yes & Yes                                        & Yes & Yes & Yes & Yes & Present Work            
\end{tabular}
\end{adjustbox}
\caption{Comparison between present work and previously existed studies based on various factors}
\label{tab:comparison}
 \end{table}

\section{Algorithm}

In this paper, we follow the idea presented in \cite{kebaier2018coupling} to develop our coupling algorithm, which we will refer to as AISML2R:
\begin{eqnarray}
\label{oneeq:12}
\mathbb{E}\left[P\left(X_{T}^{n_{L}}\right)\right]
&=&\mathbb{E}\left[P(X_{T}^{n_{1}}(\theta_{1}))\mathcal{J}^{-}(W^1_T,\theta_{1})\right]\nonumber\\
&+&\sum\limits_{l=2}^{L} \widetilde{W}_{l}\mathbb{E}\left[(P\left(X_{T}^{n_{l}}(\theta_l)\right)-P\left(X_{T}^{n_{l-1}}(\theta_{l})\right)
\mathcal{J}^{-}(W^l_T,\theta_{l})\right].
\end{eqnarray}
By applying Monte Carlo method to each level $l$ with samples $N_{l}$ in equation \eqref{oneeq:12}, we get,
\begin{eqnarray}
\label{oneeq:13}
\mathbf{J}_{\pi}^{N,\theta}\left(\theta_{1}, \dots,\theta_{L}\right)&=& \frac{1}{N_{1}}\sum\limits_{k=1}^{N_{1}}P\left(X_{T,\theta_1^{k-1}}^{n_{1},k}\right)\mathcal{J}^{-}(W^{1,k}_{T},\theta_{1}^{k-1}) \nonumber\\
&+&\sum\limits_{l=2}^{L}\frac{\widetilde{W}_{l}}{N_{l}}\sum\limits_{k=1}^{N_{l}}\left(P\left(X_{T,\theta_l^{k-1}}^{n_{l},k}\right)-P\left(X_{T,\theta_l^{k-1}}^{n_{l-1},k})\right)\mathcal{J}^{-}(W^{l,k}_{T},\theta_{l}^{k-1}\right).
\end{eqnarray}
In the above estimator, the value of $\theta^{k}_l$ on level $l$, is calculated using the values of the payoff, and $\theta^{k-1}_l$ obtained in the previous path simulation, using the stochastic approximation algorithm studied in Section \ref{stochastic_algo} below, thus describing the adaptive nature of the algorithm.

\subsection{Optimization Problem}

In this subsection, we deal with the problem of finding the optimal value of $\theta_{l}$, to minimize the variance on level $l$. For the sake of brevity, we let,
\[Y(h) \coloneqq \left(\frac{h}{M}\right)^{\frac{-\beta}{2}}\left(P\left(X^{h/M}_{T}\right)-P\left(X^{h}_{T}\right)\right)~\text{and}~Y_{l}\coloneqq Y\left(\frac{\mathbf{h}}{M^{l-2}}\right).\] 
Further let, 
\[Z_{l}=P\left(X^{h/M^{l-1}}_{T}\right)-P\left(X^{h/M^{l-2}}_{T}\right)~\text{and}~Z_{1}=Y(\mathbf{h})=P(X_{T}^{n_1}).\]
We start our discussion by recalling the Central Limit Theorems for ML2R from \cite{giorgi2017limit},

\begin{theorem}[Central Limit Theorem, $\beta >1$]\label{Central Limit Theorem}
Assume \textbf{strong error} for $\beta>1$ and that $(Y(h))_{h\in\mathcal{H}}$ is $L^2-$uniformly integrable. We set:
\[\sigma_{1}^{2}=\frac{1}{\Sigma}\frac{Var(Y_{\pmb{h}})}{Var(Y_{0})\left(1+\lambda \pmb{h}^{\frac{\beta}{2}}\right)}~\textit{and}~\sigma_{2}^{2}=\frac{1}{\Sigma}\frac{\pmb{h}^{\frac{\beta}{2}} \sum\limits_{l\geq2}M^{\frac{1-\beta}{2}(l-1)}Var(Y_l)}{\sqrt{Var(Y_0)V_1}\underline{C}_{M,\beta}},\]
with,
\[\Sigma=\left[1+\lambda \pmb{h}^{\frac{\beta}{2}}\left(1+\overline{C}_{M,\beta}\frac{M^{\frac{1-\beta}{2}}}{1-M^{\frac{1-\beta}{2}}}\right)\right],~\underline{C}_{M,\beta}=\frac{1+M^{\frac{\beta}{2}}}{\sqrt{1+M^{-1}}}, ~\text{and}~\overline{C}_{M,\beta}=\left(1+M^{\frac{\beta}{2}}\right)\sqrt{1+M^{-1}}.\]
Assuming \textbf{weak error} for $\Bar{L} \geq 1$, we have, 
\[\frac{\mathbf{J}_{\pi}^N(\epsilon)-\mathbf{J}_0}{\epsilon} \xrightarrow{\mathcal{L}}\mathcal{N}(0,\sigma_{1}^{2}+\sigma_{2}^{2}),~\text{as} \epsilon \rightarrow 0.\]
\end{theorem}

\begin{theorem}[Central Limit Theorem, $0<\beta \leq 1$]\label{Central Limit Theorem_02}
Assume \textbf{strong error} for $\beta \in (0,1]$ and that $(Y(h))_{h\in\mathcal{H}}$ is $L^2-$uniformly integrable. Assume furthermore $\displaystyle{\lim_{h\rightarrow 0}\lVert Y(h) \rVert_{2}^{2}=v_{\infty}(M,\beta)}$, we set:
\[\sigma^{2}=
\begin{cases}
v_{\infty}(M,\beta)\left(1+M^{\beta/2}\right)^{-2}V_1^{-1},~\text{if}~2\alpha>\beta, \\
\left(v_{\infty}(M,\beta)-c_1^2(1-M^{\beta/2})^2\right)\left(1+M^{\beta/2}\right)^{-2}V_1^{-1},~\text{if}~2\alpha=\beta.
\end{cases}\]
Assuming \textbf{weak error} for $\Bar{L} \geq 1$, we have,
\[\frac{\mathbf{J}_{\pi}^N(\epsilon)-\mathbf{J}_{0}}{\epsilon} \xrightarrow{\mathcal{L}}\mathcal{N}(0,\sigma^{2}),~\text{as}~\epsilon \rightarrow 0.\]
\end{theorem}
Based of the two theorems stated above we develop the optimization algorithm for $\beta>1$ and $\beta \in (0,1]$.

\subsubsection{$\text{Case I:} \beta > 1$}
In Theorem \ref{Central Limit Theorem}, let
\[\sigma_{1}^{2}=k_{1}\text{Var}(Y_{\mathbf{h}})~\text{and}~\sigma_{l}^{2}=k_{2} M^{\frac{1-\beta}{2}(l-1)}\text{Var}(Y_{l}),\]
where,
\[k_{1}=\frac{1}{\Sigma}\frac{1}{Var(Y_{0})(1+\lambda \mathbf{h}^{\frac{\beta}{2}})}~\text{and}~k_{2}=\frac{\mathbf{h}^{\frac{\beta}{2}}}{\Sigma \sqrt{Var(Y_{0})V_{1}}\underline{C}_{M,\beta}},
\] 
and therefore,
\[\sigma_{2}^{2}=\sum\limits_{l\geq2}k_{l}\text{Var}(Y_{l})=\sum\limits_{l\geq2} \sigma_{l}^{2}.\]
From the practical point of view, it is necessary to use the truncated version of the above summation. In the thorough study carried out in \cite{giorgi2017limit}, it was proven that for $\beta>1$,
\[\lim_{L(\epsilon)\rightarrow \infty}\sum\limits_{l=2}^{L(\epsilon)}|\widetilde{W}_{l}^{L(\epsilon)}|M^{\frac{1-\beta}{2}(l-1)}\text{Var}(Y_{l})=\sum\limits_{l=2}^{\infty}M^{\frac{1-\beta}{2}(l-1)}\text{Var}(Y_{l}).\] Therefore, owing to the above result and motivated by the analysis pertaining to the Central Limit Theorem carried out in \cite{giorgi2017limit}, we can formulate the problem for $l=1,\dots,L(\epsilon)$ as,
\begin{equation}
\label{oneeq:14}
\begin{aligned}
\theta_{l}^{*}&={\arg\min}_{\theta_{l} \in \mathbb{R}}\sigma_{l}^{2}\\
&={\arg\min}_{\theta_{l} \in \mathbb{R}} k_{2} M^{\frac{1-\beta}{2}(l-1)} \abs{\widetilde{W}_l^{L(\epsilon)}} \text{Var}\left[\left(Y^{\theta}_{l}\mathcal{J}^{-}(W^{l}_{T},\theta_l)\right)^{2}\right]\\
&= {\arg\min}_{\theta_{l} \in \mathbb{R}} k_{l} \mathbb{E}\left[\left(\left(P\left(X^{h/M^{l-1}}_{T,\theta}\right)-P\left(X^{h/M^{l-2}}_{T,\theta}\right)\right)\mathcal{J}^{-}(W^{l}_{T},\theta_l)\right)^{2}\right],
\end{aligned}
\end{equation}
where,
\begin{equation}
\label{oneeq:15}
k_{l}=k_{2} M^{\frac{1+\beta}{2}(l-1)}\abs{\widetilde{W}_{l}^{L(\epsilon)}} \mathbf{h}^{-\beta}.
\end{equation}
Using the Girsanov's theorem, we can see that,
\begin{equation}
\label{oneeq:16}
\mathbb{E}\left[\left(\left(P\left(X^{h/M^{l-1}}_{T,\theta}\right)-P\left(X^{h/M^{l-2}}_{T,\theta}\right)\right)\mathcal{J}^{-}(W^{l}_{T},\theta_l)\right)^{2}\right]= \mathbb{E}\left[\left(P\left(X^{h/M^{l-1}}_{T}\right)-P\left(X^{h/M^{l-2}}_{T}\right)\right)^{2}\mathcal{J}^{+}(W^{l}_{T},\theta_{l})\right],
\end{equation}
where,
\[\mathcal{J}^{+}(W^{l}_{T},\theta_l)=e^{-\innerproduct{W_{T}^{l}}{\theta_{l}}+\frac{1}{2}\lVert\theta_{l}\rVert^{2}T}.\]
Similarly for $l=1$, we have,
\begin{equation}
\label{oneeq:17}
\theta_{1}^{*}={\arg\min}_{\theta_{1} \in \mathbb{R}} k_{1}\mathbb{E}\left[\left(Y^{2}(\mathbf{h}) \mathcal{J}^{+}(W^{1}_{T},\theta_{1}) \right)\right].
\end{equation}

\subsubsection{Case II: $\beta \in (0,1]$}

Based on Theorem \ref{Central Limit Theorem_02}, we define $v_{\infty}^{l}$ be the level $l$ approximation of $v_{\infty}$,
where, 
\[v_{\infty}^{l}=\bigg\lVert Y\left(\frac{h}{M^{l-2}}\right) \bigg \rVert_{2}^{2}.\] 
Therefore, we have,
\begin{equation}
\sigma_{l}^{2}=
\begin{cases}
v_{\infty}^l(M,\beta)\left(1+M^{\beta/2}\right)^{-2}V_1^{-1},~\text{if}~2\alpha>\beta \\
\left(v_{\infty}^l(M,\beta)-c_1^2(1-M^{\beta/2})^2\right)\left(1+M^{\beta/2}\right)^{-2}V_1^{-1},~\text{if}~2\alpha=\beta. 
\label{oneeq:18}
\end{cases}
\end{equation}
We follow the similar line of development as for  $\beta>1$ to formulate the problem for $\beta \in (0,1]$. As one can observe from \eqref{oneeq:18} that, in order to minimize $\sigma^2_l$ on level $l$, we only need to minimize $v_{\infty}^l$. Therefore, we formulate the optimization problem as follows ,
\begin{equation}
\label{oneeq:19}
\theta_{l}^{*}=\arg\min_{\theta \in \mathbb{R}^q}\mathbb{E}\left[Y_{l}^{2}\mathcal{J}^+(W_{T}^{l},\theta_l)\right].
\end{equation}

\subsection{Stochastic Algorithm}
\label{stochastic_algo}

In this section we outline the stochastic approximation algorithm to estimate $\theta_l^*$. Based on the discussion in the previous section, we perform the discretization, such as Euler or Milstein, of the underlying SDE to approximate $Y_{l}$'s. Whereas for stochastic approximation we resort to the Robbins-Monro algorithm, to approximate the value of $\theta_{l}^{*},~l=1,\dots,L$. The aim is to construct a sequence of $(\theta_{l}^{n})_{n \in \mathbb{N}}$, such that, $\displaystyle{\lim_{n\rightarrow\infty}\theta_{l}^{n}=\theta_{l}^{*}}$.

Let $\Theta \subset \mathbb{R}^{q}$ be a compact convex subset such that $\theta \in \text{int}(\Theta)$. Then the sequence is constructed recursively as follows,
\begin{equation}
\label{oneeq:20}
\theta_{l}^{n+1}=\text{Proj}_{\Theta}\bigg[\theta_{l}^{n} - \gamma_{n+1}G_{l}(\theta_{l}^{n},Y_{l},W_{T,l})\bigg],
\end{equation}
where $\displaystyle{\text{Proj}_{\Theta}(\theta)=\min_{\theta \in \Theta}\abs{\theta-\theta_{0}}}$.
Also $(\gamma_n)_{n \geq 1}$ is a decreasing sequence of positive real numbers satisfying,
\begin{equation}
\label{oneeq:21}
\sum\limits_{n=1}^{\infty} \gamma_{n}=\infty~\text{and}~\sum\limits_{i=1}^{\infty}\gamma_{n}^{2} < \infty.
\end{equation}
where for $\beta > 1$,
\begin{equation}
\label{oneeq:22}
G_{l}(\theta_{l},Y_{l},W^{l}_{T})=
\bigg \{
\begin{array}{ l l }
\left(\theta_{l}T-W^{l}_{T}\right)k_{l}Z_{l}^2 \mathcal{J}^{+}(W_{T}^{l},\theta_l),&\text{for}~l = 2,\dots,L,\\
\left(\theta_{1}T-W^1_{T}\right)k_{1} Z_{1}^{2}(\mathbf{h})\mathcal{J}^{+}(W_T^1,\theta_1), &\text{for}~l=1,
\end{array}
\end{equation}
and for $\beta \in (0,1]$,
\begin{equation}
\label{oneeq:23}
G_{l}(\theta_{l},Y_{l},W^{l}_{T})=
\bigg \{
\begin{array}{ l l }
\left(\theta_l T-W^{l}_{T}\right)Y_{l}^2 \mathcal{J}^+(W_T^l,\theta_l),&\text{for}~l = 2,\dots,L,\\
\left(\theta_1 T-W^1_{T}\right)
Y^2(\mathbf{h})\mathcal{J}^+(W_T^1,\theta_1), &\text{for}~l = 1.
\end{array}
\end{equation}
It may be noted that in this paper, we use the constrained version of the Robbins-Monro algorithm in order to estimate the approximate value of $\theta_l^*$ for various level of resolutions.

\section{Analysis}
In this section we study our algorithm in detail. We start our discussion by establishing  the existence and uniqueness of the $\theta_l^*$ for various level of resolution. Further we carry out the asymptotic analysis of our AISML2R estimator, thereby establishing the Strong Law of Large Numbers.

\subsection{Existence and Uniqueness of $\theta_l^*$.}
\label{section_existanduni}

Before, establishing the existence and uniqueness of $\theta_{l}^{*}$, we recall results related to the weights $(\widetilde{W}_{l})_{l=1,\dots,L}$, from \cite{giorgi2017limit} necessary for our study.
Accordingly, let us define,
\[a_{l} \coloneqq \frac{1}{\prod_{1\leq k \leq (l-1)}\left(1-M^{-k\alpha}\right)},~l=1,\dots,L,\]
\[b_{l} \coloneqq (-1)^l\frac{M^{\frac{-\alpha l (l+1)}{2}}}{\prod_{1\leq k \leq (l-1)}\left(1-M^{-k\alpha}\right)},~l=0,\dots,L.\]
The following result was proved in \cite{giorgi2017limit},
\begin{lemma}
\label{onelemma:1}
Let $\alpha > 0$. and the associated weights $(\widetilde{W}_l)_{l = 1,\dots,L}$, be as given in 
\eqref{oneeq:6}.
\begin{enumerate}[(a)]
\item $\lim_{l \rightarrow+\infty}a_{l}=a_{\infty} < +\infty$ and $\displaystyle{\sum\limits_{l= 0}^{+\infty}\lvert b_{l} \rvert=\widetilde{B}_{\infty} < +\infty}$.
\item The weights $\widetilde{W}_{l}$ are uniformly bounded,
\[\forall L \in \mathbb{N}^{*}, \forall l \in \{1,\dots,L\},~\lvert \widetilde{W}_{l} \rvert \leq a_{\infty}\widetilde{B}_{\infty}.\]
\item For every $\gamma > 0$,
\[\lim_{L \rightarrow+\infty}\sum\limits_{l = 2}^{L}\lvert \widetilde{W}_{l}\rvert M^{-\gamma(l-1)}=\frac{1}{M^{\gamma}-1}.\]
\item Let ${v_{j}}_{j\geq1}$ be a bounded sequence of positive real numbers. Let $\gamma \in \mathbb{R}$ and assume that $\lim_{j \rightarrow+\infty}v_{j}=1$ when $\gamma > 0$. Then the following limits hold:
\[\sum\limits_{l=2}^{L}\lvert\widetilde{W}_{l}\rvert M^{\gamma(l-1)}v_{l} \sim
\begin{cases}
\sum\limits_{l\geq2}M^{\gamma(l-1)}v_l < +\infty,~\text{for}~\gamma<0   \\
R,~\text{for}~\gamma=0,~\text{as}~L \rightarrow +\infty\\
M^{\gamma L}a_{\infty}\sum\limits_{l\geq1}\lvert \sum\limits_{k=0}^{l-1}b_{k}\rvert M^{-\gamma l},
\textit{for}~\gamma>0.    
\end{cases}\]
\end{enumerate}
\end{lemma}
With above results at our disposal, we present a series of results establishing the existence and uniqueness of the optimal parameter $\theta_l^*$ on various level of resolution. For the most part the proof follows line of argument similar to that presented in \cite{kebaier2018coupling,alaya2015importance,arouna2004adaptative}. We start our discussion bu proving the following lemma, necessary for the existence and uniqueness results.

\begin{lemma}
\label{onelemma:2}
Let $p\geq2$ and assume the $\mathcal{L}^p$-strong error rate assumption, \textit{i.e.},
\begin{equation}
\label{oneeq:24}
\exists \beta > 0, V_{1}^{(p)}>0, \bigg\lVert P(X_{T}^{h})-P(X_{T})\bigg\rVert_{p}^{p}= \mathbb{E}\left[\bigg\lvert P(X_{T}^{h})-P(X_{T})\bigg\rvert^{p}\right]\leq V_{1}^{(p)}h^{\beta p/2}, h \in \mathcal{H}.
\end{equation}
Then, $\mathbb{E}\left[\sup_{\lvert \theta \rvert \leq c}\bigg \lvert G_{l}(\theta,Y_l,W_T) \bigg \rvert \right] < \infty$ for $l=2,\dots,L$ and some constant $c > 0 \in \mathbb{R}$. 
\end{lemma}
\begin{proof} 
We prove the above lemma for $\beta > 1$. The proof for $\beta \in (0,1]$ follows a similar line of argument and is relatively easy to prove. Consider,
\[\bigg \lvert G_l(\theta,Z_l,W_T) \bigg \rvert=\bigg\lvert(\theta T-W_{T})k_{l} Z_{l}^{2} e^{-\innerproduct{\theta}{W_{T}}+\frac{1}{2}\lvert \theta \rvert^{2}T} \bigg \rvert.\]
Then, for $c>0$,
\[\begin{aligned}
\sup_{\lvert \theta \rvert \leq c}\bigg \lvert G_{l}(\theta, Z_{l},W_{T}) \bigg \rvert &=\sup_{\lvert \theta \rvert \leq c}\bigg \lvert (\theta T-W_{T})k_{l}Z_{l}^{2} e^{-\innerproduct{\theta}{W_{T}}+\frac{1}{2}\lvert \theta \rvert^{2}T} \bigg \rvert,\\
&\leq (cT+\lvert W_{T}\rvert)k_{l} Z_{l}^{2} e^{c\lvert W_{T}\rvert+\frac{1}{2}\lvert c \rvert^{2}T}.
\end{aligned}\]
Taking expectation on both sides and applying H\"older's inequality for all $p\geq2$, we get,
\[\mathbb{E}\left[\sup_{\lvert \theta \rvert \leq c}\bigg \lvert G_{l}(\theta,Y_{l},W_{T}) \bigg \rvert \right] \leq e^{\frac{c^{2}}{2}T}\bigg \lVert e^{c\lvert W_{T}\rvert}(cT+\lvert W_{T} \rvert) \bigg\rVert_{\frac{p}{p-1}}\bigg \lVert k_{l} Z_{l}^{2} \bigg \rVert_{p}.\]
It is clear that $\bigg \lVert e^{c\lvert W_{T}\rvert}(cT+\lvert W_{T} \rvert) \bigg\rVert_{\frac{p}{p-1}}$ is bounded. As for $\lVert k_{l} Z_{l}^{2}\rVert_{p}$ we fallback to the $\mathcal{L}^{p}$-\textit{strong error} assumption. Therefore, we have,
\begin{equation}
\label{oneeq:25}
\lVert k_{l}Z_{l}^{2}\rVert_{p}=k_{l}\lVert Z_{l} \rVert_{2p}^2 \leq K(M,\beta,h,p)k_{l} M^{-\beta(l-1)},
\end{equation}
where,
\begin{equation}
\label{oneeq:26}
K(M,\beta,h,p)=\left(V_{1}^{(2p)}\right)^{\frac{1}{p}}\left(1+M^{\frac{\beta}{2}}\right)^{2}h^{\beta}.
\end{equation}
The boundedness of $k_{l} M^{-\beta(l-1)}$ can be derived from the definition of $k_{l}$ and from \textit{(b)} and \textit{(c)} of Lemma \ref{onelemma:1}. Hence, from the above analysis, we can conclude that,  $\mathbb{E}\left[\sup_{\lvert \theta \rvert \leq c}\bigg \lvert G_l(\theta,Y_{l},W_{T}) \bigg \rvert \right]$ is bounded for $l=2,\dots,L$.
\end{proof}

We state now, the result pertaining to the existence and uniqueness of optimal parameters on various level of resolutions.
 
\begin{proposition}
\label{oneprop:1}
Suppose $b$ and $\sigma$ satisfies assumption (\ref{onecond:A_1}). Let $P$ be such that,
\[\mathbb{P}(X_{T} \notin D_{P})=0,~\text{where}~D_{P}=\{x \in \mathbb{R}| P~\text{is differentiable at}~x\}.\] 
Further, if $Z_l$ satisfies the following conditions,
\begin{enumerate}[(1)]
\item $\mathbb{P}(Z_{l} \neq 0) > 0$. 
\item $\lVert Z_{l}^{2} \rVert_{p} < +\infty~\text{for some}~p >1$.
\end{enumerate}
Then the function $\theta \rightarrow \mathbf{v}_l(\theta)$ is $\mathcal{C}^2$ and strictly convex with $\nabla_{\theta} \mathbf{v}_l(\theta)=\mathbb{E}[G_{l}(\theta,Z_{l},W_{T})]$ for all $l \in \mathbb{N}$. Moreover, there exists an unique  $\theta^*\in \mathbb{R}^q$ such that $\min_{\theta \in \mathbb{R}^q}\mathbf{v}_l(\theta)=\mathbf{v}_l(\theta^{*})$.
\end{proposition}

\begin{proof}
To prove the proposition, we refer to the proof in \cite{alaya2015importance,kebaier2018coupling}, in our context. Here we discuss the proof for $l\geq 2$ and $\beta >1$. For $l=1$ and $\beta \in (0,1]$ the proofs are relatively easy and can be reproduced in the similar way. To begin with, one can observe that $\displaystyle{\theta \rightarrow k_{l} (Z_{l})^{2} e^{-\innerproduct{\theta}{W_{T}}+\frac{1}{2}\lvert \theta \rvert^{2}T}}$ is a continuously infinitely differentiable function with respect to $\theta$, and $\displaystyle{\frac{\partial}{\partial \theta^{j}}(k_{l}(Z_{l})^{2} e^{-\innerproduct{\theta}{W_T}+\frac{1}{2}\lvert \theta \rvert^{2}T})=k_{l}(Z_{l})^{2} (\theta^{j}T-W^{j}_{T}) e^{-\innerproduct{\theta}{W_{T}}+\frac{1}{2}\lvert \theta \rvert^{2}T}}$. Therefore, the first derivative of the map $\displaystyle{\theta \rightarrow k_{l} (Z_{l})^{2} e^{-\innerproduct{\theta}{W_{T}}+\frac{1}{2}\lvert \theta \rvert^{2}T}}$ is equal to $G_{l}(\theta,Z_{l},W_{T})$. As we have already seen in Lemma \ref{onelemma:2} that the $\sup_l\mathbb{E}[\sup_{\lvert \theta \rvert \leq c}\lvert G_{l}(\theta,Z_{l},W_{T})\rvert]$ is bounded, therefore by Lebesgue's theorem we can conclude that $\mathbf{v}_{l}(\theta)$ is $\mathcal{C}^{1}(\mathbb{R}^{q})$, with $\nabla_{\theta}\mathbf{v}_{l}(\theta)=\mathbb{E}[G_{l}(\theta,Z_{l},W_{T})]$ for all $l \in \mathbb{N}$. A similar line of argument also proves that $\mathbf{v}_{l}(\theta)$ is $\mathcal{C}^{2}(\mathbb{R}^{q})$. The Hessian of $\mathbf{v}_{l}(\theta)$ is given as follows,
\[Hess(\mathbf{v}_{l}(\theta))=\mathbb{E}\left[\left((\theta T-W_{T})(\theta T-W_{T})^{*}+ TI_{q}\right)k_{l}(Z_{l})^2e^{-\innerproduct{\theta}{W_{T}}+\frac{1}{2}\lvert \theta \rvert^{2}T}\right].\]
Since, $\mathbb{P}(Z_{l} \neq 0) >0$, therefore, for all $u \in \mathbb{R}^{q} \setminus \{0\}$, $u^{T} Hess(\mathbf{v}_{l}(\theta))u >0$. Hence, we can conclude that $\mathbf{v}_{l}(\theta)$ is strictly convex. As a consequence, there exists a minimum $\theta^{*} \in \mathbb{R}^{q}$, such that $\displaystyle{\min_{\theta \in \mathbb{R}^{q}}\mathbf{v}_{l}(\theta)=\mathbf{v}_{l}(\theta^{*})}$. Further,the unique minimum is attained for a finite value of $\theta$, it is enough to prove that $\displaystyle{\lim_{\lvert \theta \rvert \to +\infty}\mathbf{v}_l(\theta)=+\infty}$. This can be proved using Fatou's lemma as, 
\[+\infty=\mathbb{E}\left[\liminf_{\abs{\theta}\rightarrow +\infty}k_{l}(Z_l)^2e^{-\innerproduct{\theta}{W_{T}}+\frac{1}{2}\lvert \theta \rvert^{2}T}\right] \leq \liminf_{\abs{\theta}\rightarrow +\infty}\mathbb{E}\left[k_{l}(Z_{l})^2e^{-\innerproduct{\theta}{W_{T}}+\frac{1}{2}\lvert \theta \rvert^{2}T}\right].\]
This completes the proof.
\end{proof}

As the above proposition guarantees the existence and uniqueness of the $\theta_{l}^{*}$ for various level of resolutions, we now state the main result proving the convergence of the sequence $\theta_{l}^{k} \xrightarrow{a.s.} \theta_{l}^{*}$, as $k \rightarrow \infty$, under the assumptions of the Lemma \ref{onelemma:2} and Proposition \ref{oneprop:1}.

\begin{theorem}
\label{onetheorem:3}
If $\theta_{l}^{*}=\arg\min_{\theta_{l} \in \mathbb{R}^{q}} \mathbf{v}_{l}(\theta)$ is such that, $\nabla_{\theta}\mathbf{v}_l(\theta_{l}^{*})=0$ and $\theta_{l}^{*} \in \Theta$, then $\theta_{l}^{k} \xrightarrow{a.s.} \theta_{l}^{*}$, as $k \rightarrow \infty$.
\end{theorem}
\begin{proof}
To prove the above result we follow the assertion made in Theorem A.1 \cite{laruelle2013optimal} that suggest that in order to prove $\theta_{l}^{k} \rightarrow \theta_{l}^{*}$, where the sequence $(\theta_{l}^{k})_{k\geq1}$ is constructed through a constrained version of the Robbins Monro, we need to verify two conditions, namely,
\begin{enumerate}[(1)]
\item $\forall~\theta \neq \theta_{l}^{*},~ \innerproduct{\nabla_{\theta}\mathbf{v}_{l}(\theta)}{\theta-\theta_{l}^{*}} > 0$.
\item Non explosion condition: $\exists~C>0$ such that $\forall \theta \in \Theta,~ \mathbb{E}\left[\abs{G_{l}(\theta,Z_{l},W_{T})}^2\right] < C(1+\abs{\theta^2})$.
\end{enumerate}
As we know, $\displaystyle{\nabla_{\theta}\mathbf{v}_{l}(\theta_{l}^{*})=0}$ and in the previous proposition it was proven that $v_l$ is convex, therefore as a consequence, we prove that,
\begin{equation}
\label{oneeq:27}
\forall \theta \neq \theta_{l}^{*},~ \innerproduct{\nabla_{\theta}\mathbf{v}_{l}(\theta)}{\theta-\theta_{l}^{*}} > 0.
\end{equation}
For the non-explosion condition, we use the Cauchy-Schwarz inquality,
\begin{equation}
\label{oneeq:28}
\mathbb{E}\left[\abs{G_{l}(\theta,Z_{l},W_{T})}^2\right] \leq e^{\abs{\theta}^{2}T}\left(\mathbb{E}[k_{l}^{4} Z_{l}^{8}]\right)^{\frac{1}{2}}\left(\mathbb{E}[\abs{e^{-\innerproduct{\theta}{W_{T}}}(\theta T-W_{T})}^{4}]\right)^{\frac{1}{2}}.
\end{equation}
Under the assumption and following the similar line of argument of Lemma \ref{onelemma:2}, it is easy to prove that there exists a constant $C>0$, such that,  
\begin{equation}
\label{oneeq:29}
\mathbb{E}\left[\abs{G_{l}(\theta,Z_{l},W_{T})}^{2}\right] \leq e^{\abs{\theta}^{2}T}C\left(\mathbb{E}[\abs{e^{-\innerproduct{\theta}{W_{T}}}(\theta T-W_{T})}^4]\right)^{\frac{1}{2}}.
\end{equation}
Further using the fact the $\theta \in \Theta$, we can deduce that,
\begin{equation}
\label{oneeq:30}
\sup_{\theta \in \Theta}\mathbb{E}\left[\abs{G_{l}(\theta,Z_{l},W_{T})}^{2}\right] < \infty.
\end{equation} which in turns concludes the non explosion condition. This proves the almost sure convergence of $\theta_{l}^{k}$ to $\theta_{l}^{*}$ as $k \rightarrow \infty$.
\end{proof}

\subsection{Main Result}
\label{section_mainresult}

\begin{theorem}
\label{onetheore:4}
Let $(\theta_{l}^{k})_{k \geq 0} \subset \Theta $ be a family of sequence converging to $\theta_{l}^{*}\in \Theta $ as $k\rightarrow \infty$. Moreover for $p\geq2$ assume \textbf{weak error} for all $L\geq1$ and $P(X_{T}) \in \mathbf{L}^{p}$. Furthermore assume the  $\mathbf{L}^{p}$-\textbf {strong error} rate assumption,\textit{i.e.,}
\begin{equation}
\label{oneeq:31}
\exists \beta >0, V_{1}^{(p})>0, \bigg\lVert P(X_{T}^{h})-P(X_{T})\bigg\rVert_{p}^{p}=\mathbb{E}\left[\bigg\lvert P(X_{T}^{h})-P(X_{T})\bigg\rvert^{p}\right]\leq V_{1}^{(p)}h^{\beta p/2}, h \in \mathcal{H}.
\end{equation}
If $(\epsilon_{k})_{k\geq1}$ is a sequence of positive real numbers such that $\displaystyle{\sum\limits_{k\geq1}\epsilon^{p}_k < +\infty}$, then the AISML2R estimator satisfies,
\begin{equation}
\label{oneeq:32}
\mathbf{J}_{\pi}^{N,\theta} \xrightarrow{a.s.} \mathbf{J}_{0}.
\end{equation}
\end{theorem}
We will assume the following notation,
\[\widetilde{\mathbf{J}}_{\theta,\pi}^{1} \coloneqq \frac{1}{N_{1}}\sum\limits_{k=1}^{N_{1}}P(X_{T,\theta^{k-1}_{1}}^{n_{1},k})\mathcal{J}^{-}(W_T^{1,k},\theta^{1,k-1}_1)-\mathbb{E}\left[P(X_{T}^{n_1})\right]~\text{and}~\widetilde{\mathbf{J}}_{\theta,\pi}^{2}\coloneqq\sum\limits_{l=2}^L\frac{\widetilde{W}_l}{N_l}\sum\limits_{k=1}^{N_l}\widetilde{Y}_{l,\theta^{k-1}_l}^{k},\]
where we set,
\begin{equation}
\label{oneeq:33}
\widetilde{Y}_{l,\theta_l}=\left(P(X^{n_{l}}_{T,\theta_{l}})- P(X^{n_{l-1}}_{T,\theta_l})\right)\mathcal{J}^{-}(W_l,\theta_l)-\mathbb{E}\left[P(X^{n_{l}}_{T})-P(X^{n_{l-1}}_{T})\right],
\end{equation}
and 
\begin{equation}\
\label{oneeq:34}
\widetilde{Y}_{1,\theta_1}=\left(P(X^{n_{1}}_{T,\theta_1})\right)\mathcal{J}^{-}(W_1,\theta_{1})-\mathbb{E}\left[P(X^{n_{1}}_{T})\right].
\end{equation}
Therefore, we have,
\[\mathbf{J}_{\pi}^{N,\theta}-\mathbf{J}_{0}=\widetilde{\mathbf{J}}_{\theta,\pi}^{1} +\widetilde{\mathbf{J}}_{\theta,\pi}^{2}+\mathbb{E}[P(X_{T}^{L})]-\mathbf{J}_{0}.\]
A thorough analysis carried out in Section 4.2 of \cite{giorgi2017limit} shows that the last term in the equation converges to zero as $\epsilon \rightarrow \infty$. We start our discussion by proving the following lemma.

\begin{lemma}
\label{onelemma:3}
Let $p\geq2$ and $\lvert \theta \rvert \leq c$. Then there exist a positive constant $K_{1}(M,\beta,p,c)$ such that,
\[\lVert \widetilde{Y}_{l,\theta} \rVert_{p}^{p} \leq K_{1}(M,\beta,p,c)M^{-\beta p(l-1)/2},~l=2,\dots,L.\]
\end{lemma}
\begin{proof}
By Minkowski's Inequality, we have,
\begin{equation}
\label{oneeq:35}
\begin{aligned}
\left(\mathbb{E}\left[|\widetilde{Y}_{l,\theta})|^{p}\right]\right)^{1/{p}}&\leq \biggl\lVert \left(P(X^{n_{l}}_{T,\theta})-P(X^{n_{l-1}}_{T,\theta})\right)\mathcal{J}^{-}(W^{l}_{T},\theta)\biggr\rVert_{p}+ \biggl\lvert\mathbb{E}\left[P(X^{n_{l}}_{T})-P(X^{n_{l-1}}_{T})\right]\biggr\rvert,\\
&\leq \underbrace{\biggl\lVert \left(P(X^{n_{l}}_{T,\theta})- P(X^{n_{l-1}}_{T,\theta})\right)\mathcal{J}^{-}(W^{l}_{T},\theta)\biggr\rVert_{p}}_{\text{I}}+\underbrace{ \biggl\lVert\left[P(X^{n_{l}}_{T})-P(X^{n_{l-1}}_{T})\right]\biggr\rVert_{p}}_{\text{II}}.
\end{aligned}
\end{equation}
In order to bound $(\textit{I})$ of \eqref{oneeq:35}, we apply the Holder's Inequality,
\begin{equation}
\label{oneeq:36}
\begin{aligned}
\biggl\lVert \left(P(X^{n_{l}}_{T,\theta})-P(X^{n_{l-1}}_{T,\theta})\right)\mathcal{J}^{-}(W^{l}_{T},\theta)\biggr\rVert_{p}^{p} &= \mathbb{E}\left[\bigg\lvert \left(P(X_{T,\theta}^{n_{l}})- P(X_{T,\theta}^{n_{l-1}})\right)\mathcal{J}^{-}(W^{l}_{T},\theta)\bigg\rvert^{p}\right],\\ &=\mathbb{E}\left[\bigg\lvert \left(P(X_{T}^{n_{l}})-P(X_{T}^{n_{l-1}})\right)\bigg \rvert^{p}\left(\mathcal{J}^+(W^{l}_{T},\theta)\right)^{p-1}\right],\\
&\leq \left(\mathbb{E}\left[\bigg\lvert \left(P(X_T^{n_l})-P(X_T^{n_{l-1}})\right)\bigg \rvert^{p^2}\right]\right)^{\frac{1}{p}}\bigg\lVert
\left(\mathcal{J}^+(W^{l}_{T},\theta)\right)^{p-1}\bigg\rVert_{\frac{p}{p-1}},\\
&\leq e^{\frac{(p^{2}-1)}{2}c^{2}T}\bigg\lVert P(X_{T}^{n_{l}})-P(X_{T}^{n_{l-1}})\bigg \rVert_{p^{2}}^{p}.
\end{aligned} 
\end{equation}
Therefore, from above analysis, we have,
\begin{equation}
\label{oneeq:37}
\biggl\lVert \left(P(X^{n_{l}}_{T,\theta})-P(X^{n_{l-1}}_{T,\theta})\right)\mathcal{J}^{-}(W^{l}_{T},\theta)\biggr\rVert_{p}\leq e^{\frac{(p^{2}-1)}{2p}c^{2}T}\bigg\lVert P(X_{T}^{n_l})-P(X_{T}^{n_{l-1}})\bigg \rVert_{p^{2}}.    
\end{equation}
Further for $(II)$ of \eqref{oneeq:35}, the assumption \eqref{oneeq:31} yields,
\begin{equation}
\label{oneeq:38}
\biggl\lVert\left[P(X^{n_{l}}_{T})-P(X^{n_{l-1}}_{T})\right]\biggr\rVert_{p}^{p} \leq V_{1}^{(p)}(1+ M^{\beta/2})^{p}\mathbf{h}^{\beta p/2}M^{-\beta p(l-1)/2}.
\end{equation}
Now from \eqref{oneeq:35} and \eqref{oneeq:37}, we get,
\begin{equation}
\label{oneeq:39}
\left(\mathbb{E}\left[\lvert \widetilde{Y}_{l,\theta})\rvert^{p}\right]\right)^{1/{p}} \leq \left(1+ e^{\frac{(p^2-1)}{2p}c^2T}\right)\bigg\lVert P(X_{T}^{n_{l}})-P(X_{T}^{n_{l-1}})\bigg \rVert_{p^{2}}.
\end{equation}
Combining above inequality and \eqref{oneeq:38} yields,
\begin{equation}
\label{oneeq:40}
\lVert \widetilde{Y}_{l,\theta} \rVert_{p}^{p} \leq K_{1}(M,\beta,p,c)M^{-\beta p(l-1)/2},
\end{equation}
where,
\[K_{1}(M,\beta,p,c)=\left(1+e^{\frac{(p^2-1)}{2p}c^2T}\right)\left(V_{1}^{(p^{2})}\right)^{1/p}(1+M^{\beta/2})^{p}\mathbf{h}^{\beta p/2}.\]
\end{proof}

\begin{proposition}
\label{oneprop:2}
Let $p \geq 2$ and $\theta_l^k \in \Theta $ for $l=2,\dots,L$ and $k \in \mathbb{N}^{*}$. Then there exists a constant $K_{2}(M,\beta,p,c)$, such that,
\begin{equation}
\label{oneeq:41}
\mathbb{E}\left[\lvert \widetilde{\mathbf{J}}_{\theta,\pi}^{2,\epsilon} \rvert^{p} \right] \leq K_{2}(M,\beta,p,c)\epsilon^{p}~\text{for some constant}~c >0. 
\end{equation}
\end{proposition}
\begin{proof}
We start our discussion by observing that as $\theta_{l}^{k} \in \Theta$, there exists a positive constant $c$ such that $\abs{\theta_{l}^{k}} \leq c,~\forall k \in \mathbb{N}$. Now, we define the following filtration $\left(\mathcal{F}_{T,k}\right)_{k\geq1}$, where $\mathcal{F}_{T,k} \coloneqq \sigma(W_{t,j},~j\leq k, t \leq T)$. With this filtration, one can readily observe that $\displaystyle{\sum\limits_{k=1}^{j}\widetilde{Y}_{l,\theta_l^{k-1}}^{k}}$ is a martingale with respect to $\mathcal{F}_{T,j}$. Now consider the following definition,
\begin{equation}
\label{oneeq:42}
s_{l} \coloneqq \sum\limits_{k=1}^{N_{l}} \widetilde{Y}_{l,\theta_{l}^{k-1}}^{k},~\text{for}~l=2,\dots,L.
\end{equation}
Since $\displaystyle{\sum\limits_{k=1}^{j}\widetilde{Y}_{l,\theta_{l}^{k-1}}^{k}}$ is $\mathcal{F}_{T,j}$ martingale, therefore by Rosenthal's inequality \cite{hall2014martingale}, we have,
\begin{equation}
\label{oneeq:43}
\norm{s_{l}}_{p}^{p} \leq C_{p} \left\{\mathbb{E}\left[\sum\limits_{k=1}^{N_{l}} \mathbb{E}[(\widetilde{Y}_{l,\theta_{l}^{k-1}})^{2}|\mathcal{F}_{T,k-1}] \right]^{p/2}+ 
\sum\limits_{k=1}^{N_l}\mathbb{E}\left[\abs{\widetilde{Y}_{l,\theta_{l}^{k-1}}}^{p}\right]\right\}.
\end{equation}
Now, from the previous Lemma, one can easily conclude that,
\begin{equation}
\label{oneeq:44}
\sum\limits_{k=1}^{N_l}\mathbb{E}\left[\abs{\widetilde{Y}_{l,\theta_{l}^{k-1}}}^{p}\right]\leq N_{l}K_{1} M^{-\beta p (l-1)/2}.   
\end{equation}
As for the first term, we use Theorem A.8 of \cite{hall2014martingale}, to obtain,
\begin{equation}
\label{oneeq:45}
\begin{aligned}
\mathbb{E}\left[\sum\limits_{k=1}^{N_{l}} \mathbb{E}\left[(\widetilde{Y}_{l,\theta_{l}^{k-1}})^{2}| \mathcal{F}_{T,k-1}\right]\right]^{p/2}&\leq  A_{p} \mathbb{E}\left[\sum\limits_{k=1}^{N_{l}} (\widetilde{Y}_{l,\theta_{l}^{k-1}})^{2}\right]^{p/2}
=A_{p}\bnorm{\sum\limits_{k=1}^{N_l} (\widetilde{Y}_{l,\theta_{l}^{k-1}})^2}_{p/2}^{p/2}\\
&\leq A_{p} \left(\sum\limits_{k=1}^{N_l}\norm{(\widetilde{Y}_{l,\theta_{l}^{k-1}})^2}_{p/2}\right)^{p/2}
\leq A_{p}(N_{l})^{p/2} K_{1}M^{-\beta p (l-1)/2},
\end{aligned}
\end{equation}
where the last inequality is the consequence of the previous Lemma. Therefore, we have for $p \geq 2$,
\begin{equation}
\label{oneeq:46}
\begin{aligned}
\norm{s_l}^{p}_{p} &\leq C_{p} \left\{A_{p}(N_{l})^{p/2} K_{1}M^{-\beta p (l-1)/2}+N_{l}K_{1} M^{-\beta p (l-1)/2} \right\}
& \leq C_{p}\left\{2A_{p}(N_{l})^{p/2} K_{1}M^{-\beta p (l-1)/2}\right\}.
\end{aligned}
\end{equation}
Now, let $K_{p}\coloneqq C_{p}A_{p}K_{1}$. Therefore, we have,
\begin{equation}
\label{oneeq:47}
\mathbb{E}[\abs{s_{l}}^{p}]\leq \left(K_{p}(N_{l})^{p/2}M^{-\beta p (l-1)/2}\right).
\end{equation}
Now, consider the following,
\begin{equation}
\label{oneeq:48}
\mathbb{E}\left[\lvert \widetilde{\mathbf{J}}_{\theta,\pi}^{2} \rvert^{p}\right]=\mathbb{E}\left[\bigg \lvert \sum\limits_{l=2}^{L}\frac{\widetilde{W}_{l}}{N_{l}}\mathbf{s}_{l}\bigg \rvert^{p} \right].
\end{equation}
As one can observe that, $(s_l)_{l\geq2}$ are independent random variable in $l$, therefore, by a version of Rosenthal's inequality \cite{hall2014martingale}, we have,
\begin{equation}
\label{oneeq:49}
\begin{aligned}
\mathbb{E}\left[\bigg \lvert \sum\limits_{l=2}^{L} \frac{\widetilde{W}_{l}}{N_{l}} s_{l}\bigg \rvert^{p} \right] &\leq \delta_{p} \left\{\left(\sum\limits_{l=2}^{L}\mathbb{E}\left[\frac{\widetilde{W}_{l}}{N_{l}}s_{l}\right]^{2}\right)^{p/2}+ \sum\limits_{l=2}^{L}\mathbb{E}\left[\babs{\frac{\widetilde{W}_{l}}{N_{l}}s_{l}}^{p}\right] \right\}\\
&\leq 2\delta_{p} \left(\sum\limits_{l=2}^{L}\mathbb{E}\babs{\frac{\widetilde{W}_{l}}{N_{l}}s_{l}}^{2}\right)^{p/2}.
\end{aligned}
\end{equation}
For $p=2$, we have,
\[\mathbb{E}\left[\abs{s_{l}}^{2}\right]\leq K_{2}(N_{l})M^{-\beta(l-1)}.\]
Therefore,
\begin{equation}
\label{oneeq:50}
\begin{aligned}
\mathbb{E}\left[\bigg \lvert \sum\limits_{l=2}^{L}\frac{\widetilde{W}_{l}}{N_{l}}s_{l}\bigg \rvert^{p} \right] &\leq 2\delta_{p}\left (\sum\limits_{l=2}^{L} \frac{\abs{\widetilde{W}_{l}}^{2}}{N_{l}}K_{2}M^{-\beta(l-1)}\right)^{p/2}. 
\end{aligned}
\end{equation}
We know that, $N_{l}=\ceil{N \mu_{l}} \geq N \mu_{l}$, and as a result we have,
\[\frac{1}{N_{l}}\leq\frac{1}{N \mu_{l}},~l=1,\dots,L.\] 
Further, owing to the expression for $\mu_{l}$, we have,
\[\frac{\abs{\widetilde{W}_{l}}}{\mu_{l}} \leq \frac{1}{\lambda \mathbf{h}^{\frac{\beta}{2}}\underline{C}_{M,\beta}q^{*}}M^{\frac{(\beta+1)(l-1)}{2}},~l=2,\dots,L.\] 
Since, $\displaystyle{\sup_{l\in(1,\dots,L), L\geq 1}\abs{\widetilde{W}_{l}}\leq a_{\infty}\widetilde{B}_{\infty}}$ \cite{giorgi2017limit}, therefore, combining everything we get,
\begin{equation}
\label{oneeq:51}
\mathbb{E}\left[\bigg \lvert \sum\limits_{l=2}^{L} \frac{\widetilde{W}_{l}}{N_{l}} s_{l}\bigg \rvert^{p} \right] \leq \widetilde{K}\left(\frac{1}{N}\sum\limits_{l=2}^{L}M^{\frac{(1-\beta)(l-1)}{2}} \right)^{p/2},
\end{equation}
where, $\displaystyle{\widetilde{K}=2\delta_{p}\left(\frac{a_{\infty}\widetilde{B}_{\infty}K_{2}}{\lambda \mathbf{h}^{\frac{\beta}{2}}\underline{C}_{M,\beta}q^{*}} \right)^{\frac{p}{2}}}$. Now as a consequence of Lemma 4.5 in \cite{giorgi2017limit}, with $\overline{\epsilon} \rightarrow 0$, we have,
\begin{equation}
\label{oneeq:52}
\forall~\epsilon \in (0,\overline{\epsilon}],~\frac{1}{N}\leq \frac{2}{C_{\beta}}\epsilon^{2} 
\begin{cases}
1,&~\text{if}~\beta >1,\\
L^{-1},&~\text{if}~\beta = 1.\\ 
M^{-\frac{1-\beta}{2}L},&~\text{if}~\beta < 1.
\end{cases}
\end{equation}
Moreover, it is easy to prove that,
\begin{equation}
\label{oneeq:53}
\sum\limits_{l=2}^{L}M^{\frac{(1-\beta)(l-1)}{2}} \leq 
\begin{cases}
\frac{1}{1-M^{\frac{1-\beta}{2}}},&~\text{if}~\beta >1,\\
L,&~\text{if}~\beta =1,\\
\frac{M^{\frac{1-\beta}{2}L}}{M^{\frac{1-\beta}{2}}-1},&~\text{if}~\beta < 1.
\end{cases}
\end{equation}
With all the preceding discussion, we have,
\begin{equation}
\label{oneeq:54}
\mathbb{E}\babs{\sum\limits_{l=2}^{L}\frac{\widetilde{W}_{l}}{N_{l}}s_{l}}^{p} \leq K_{2}\left(M,\beta,p,c\right)\epsilon^{p},
\end{equation}
where,
\begin{equation}
\label{oneeq:55}
K_{2}\left(M,\beta,p,c\right)=\widetilde{K}\left(\frac{2}{C_{\beta}}\right)^{p/2}
\begin{cases}
(1-M^{\frac{1-\beta}{2}})^{-p/2},&~\text{if}~\beta >1,\\
1,&~\text{if}~\beta = 1,\\
(M^{\frac{(1-\beta)}{2}}-1)^{-p/2}, &~\text{if}~\beta <1.
 \end{cases}
\end{equation}
\end{proof}
With above results in our hand we are ready to prove the Strong Law of Large Numbers.

\begin{proof}[Proof of Theorem \ref{onetheore:4}] 
We start our proof by proving, $\displaystyle{\mathbf{\widetilde{J}}_{\theta,\pi}^{2,\epsilon_k} \xrightarrow{a.s.} 0}$ as $k \rightarrow \infty$. Clearly, as a consequence of our assumption on $(\epsilon_k)_{k\geq1}$, we have,
\[\sum\limits_{k\geq1}\mathbb{E}\left[\abs{\mathbf{\widetilde{J}}_{\theta,\pi}^{2,\epsilon_k}}^{p}\right] < +\infty.\]
Hence, by Beppo-Levi's Theorem, $\displaystyle{\sum\limits_{k\geq1}\abs{\mathbf{\widetilde{J}}_{\theta,\pi}^{2,\epsilon_k}}^{p} < +\infty ~ a.s.}$, and as an implication we have $\displaystyle{\mathbf{\widetilde{J}}_{\theta,\pi}^{2,\epsilon_k} \xrightarrow{a.s.} 0}$ as $k \rightarrow +\infty$. We now turn our attention to proving $\displaystyle{\mathbf{\widetilde{J}}_{\theta,\pi}^{1,\epsilon_k} \xrightarrow{a.s.} 0}$ as $k \rightarrow \infty$. It is clear that as $k \rightarrow \infty$, $\epsilon_k \rightarrow 0$, which in turns implies $N_{1} \rightarrow \infty$. As one can observe that, $\displaystyle{\sum\limits_{k=1}^{j} \widetilde{Y}^{1}_{k,\theta_{1}^{k-1}}}$ is $\mathcal{F}_{T,j}$ martingale. Therefore, as an application of Rosenthal's inequality for $p=2$, we have,
\begin{equation}
\label{oneeq:56}
\begin{aligned}
\mathbb{E}\left[\abs{\mathbf{\widetilde{J}}_{\theta,\pi}^{1,\epsilon_k}}^2\right] &\leq \frac{C_2}{N_1^2} \left\{\mathbb{E}\left[\sum\limits_{k=1}^{N_{1}}\mathbb{E}\left[\left(\widetilde{Y}^{1}_{k,\theta_{1}^{k-1}}\right)^{2}\bigg \lvert \mathcal{F}_{T,k-1}\right]\right]\right\},~\text{where}~C_{2}~\text{is a constant}.\\
&=\frac{C_{2}}{N_{1}^{2}} \left\{\sum\limits_{k=1}^{N_{1}}\left(\mathbb{E}\left[P(X_{T}^{n_{1}})^{2} \mathcal{J}^+(\theta_{1}^{k-1},W_{T})\right]-\left[\mathbb{E}[P(X_{T}^{n_{1}})]\right]^{2}\right)\right\},
\end{aligned}
\end{equation}
where the last equality is the consequence of $X_{T,k}^{1}$ being independent of $\mathcal{F}_{T,k-1}$ and $\theta_{k-1}^{1}$ being $\mathcal{F}_{T,k-1}$ measurable. Further, due to to Grisanov's theorem, we introduce a couple of random variable $X_{T}^{1}$ and $W_{T}$ independent of $\displaystyle{\mathcal{F}_{T}=\cup_{k\geq1}\mathcal{F}_{T,k}}$, justifying the last equality.
Further, as $\theta_{1}^{k-1} \in \Theta$, therefore there exists a $c>0$ such that $\abs{\theta_{1}^{k-1}} \leq c$ for $k\in \mathbb{N}$. As a consequence, $\displaystyle{\sup_{k\in \mathbb{N}} \abs{P(X_{T}^{n_1})^{2} \mathcal{J}^{+}(\theta_{1}^{k-1},W_{T})} \leq P(X_{T}^{n_1})^{2} e^{c\abs{W_{T}} + \frac{c^{2}}{2}T}}$. Now for $p\geq 2$,
\[\mathbb{E}\left[P(X_{T}^{n1})^{2} e^{c\abs{W_{T}}+\frac{c^{2}}{2}T}\right]\leq \bnorm{P(X_{T}^{n_{1}})^{2}}_{p} \bnorm{e^{c\abs{W_{T}}+\frac{c^{2}}{2}T}}_{\frac{p}{p-1}} < +\infty.\]
Therefore, as consequence of Theorem \ref{onetheorem:3} and Lebesgue theorem, we obtain that,
\[\lim_{k\rightarrow\infty}\mathbb{E}[P(X_{T}^{n_{1}})^{2} \mathcal{J}^{+}(\theta_{1}^{k-1},W_{T})]= \mathbb{E}[P(X_{T}^{n_{1}})^{2} \mathcal{J}^{+}(\theta_{1}^{*},W_{T})].\]
Now as the application of Cesaro's lemma in equation \eqref{oneeq:55} one can easily conclude that, $\displaystyle{\mathbf{\widetilde{J}}_{\theta,\pi}^{1,\epsilon_k} \xrightarrow{a.s.} 0}$ as $k \rightarrow +\infty$.
\end{proof}

\section{Numerical Results}
\label{section_numerical}

In this section we will illustrate the efficacy of the algorithm introduced, through a couple of examples. In order to keep the things simple we restrict ourselves to one dimensional problems, arising in the realm of mathematical finance. More specifically, we look at the results in the case of European call and lookback call options. We start our computation by determining the optimal parameters required to perform the simulations. In order to do so, we perform a pre-simulation to approximate the value of $V_{1}$, $\lambda$ and $\text{Var}(P(X_{T}^{0}))$, necessary for the computation of the optimal parameters. For computing $V_{1}$ we refer to the formula presented in \cite{lemaire2017multilevel}, \textit{i.e.,}
\begin{equation}
\label{oneeq:57}
V_{1}=\left(1+M_{\max}^{-\beta/2}\right)^{-2}\mathbf{h}^{-\beta}\norm{P\left(X_{T}^{h}\right)- P\left(X^{h/M_{\max}}_{T}\right)}_{2}^{2}.
\end{equation}
Here we set $M_{\max}=10$. As for $\text{Var}(P(X_T^0))$, we perform a small prior simulations and empirically calculate the variance. Further, the value of $\lambda$ is calculated as $\displaystyle{\lambda=\sqrt{\frac{V_{1}}{\text{Var}(P(X_{T}^{0}))}}}$. We use the above parameters to perform the simulation of ML2R. As for performing the adaptive simulation, we calculate $V_{1}^{\theta}$, $\text{Var}(P(X_{T}^{0,\theta})\mathcal{J}^{-}(W_{T},\theta)$ and $\lambda_{\theta}$, to compute the necessary parameters. To compute these structural parameters, we perform a dummy optimization procedure based on the Robbins-Monro algorithm, to calculate the value of $\theta$. We calculate the value of $V_{1}^{\theta}$ using the same formula used for the calculation of $V_{1}$. We again run a small simulation as before, to calculate the value of $\text{Var}(P(X_{T}^{0,\theta})\mathcal{J}^{-}(W_{T},\theta)$, further calculating $\displaystyle{\lambda_{\theta}= \sqrt{\frac{V_{1}^{\theta}}{\text{Var}(P(X_{T}^{0,\theta})\mathcal{J}^{-}(W_{T},\theta)}}}$. Based on these structural parameters, and using the formulas from Table \ref{T:Optimal Parameters}, we calculate the parameters to perform the adaptive simulation. In order to perform the simulation for both the adaptive and the non-adaptive algorithm, we use the Milstein scheme, unlike in the original study carried out in \cite{kebaier2018coupling}, which uses Euler-Maruyama scheme, to discretize the underlying SDE. The choice of the scheme is due to the establishment of the Central Limit Theorem in \cite{giorgi2017limit} for $\beta>1$, which provides us with the freedom to use discretization scheme of higher order. In order to determine these structural parameters, we input the value of $\epsilon$, the desired root-mean-squared error. For, the purpose of our numerical experiments, we set $\epsilon=2^{-k}$, where $k=3,\dots,9$.

It may be noted that in order to simulate in the probability space $\left(\Omega,\{\mathcal{F}_{t}^{\theta}\}_{t \geq 0},\mathbb{P}_{\theta}\right)$ \textit{i.e.,} for adaptive simulation, the calculation of $L(\epsilon)$ and $N(\epsilon)$ are rather sub-optimal. This is because the formulas associated with these calculations do not consider the number of iterations required to estimate $\theta_{l}^{*}$, on various level of resolutions. However, while using Milstein discretization in the parametric probability space, it was observed that the values for $L(\epsilon)$ and $N(\epsilon)$, calculated using the procedure described above is sufficient to achieve an accuracy comparable to that achieved by ML2R, whereas, under the Euler discretization, twice the number of paths generated through the formula are used. As we will see, even by increasing the number of sample paths the adaptive algorithm outperforms standard ML2R, thereby achieving the desired level of accuracy. We now briefly describe the numerical scheme used to perform simulation in either probability space. Consider the general one-dimensional SDE,
\begin{equation}
\label{oneeq:58}
 dX_{t}=b(X_{t},t)dt+\sigma(X_{t},t)dW_{t}.
\end{equation}
The Milstein discretization of the above equation is given as,
\begin{equation}
\label{oneeq:59}
X_{n+1}=X_{n}+b_{n}h+\sigma_{n}\Delta W_{n}+\frac{1}{2}\sigma_{n}'\sigma_{n}\left((\Delta W_{n})^{2}-h\right).
\end{equation}
In the above equation, $h$ is the uniform time-step, $b_{n}=b(X_{n},t_{n})$, $\sigma_{n}=\sigma(X_{n},t_{n})$ and $\sigma_{n}'=\sigma'(X_{n},t_{n})$, with $t_{n}\coloneqq nh$. However, under the parametric change of measure, with $\theta$ as the parameter, we have the following SDE,
\begin{equation}
\label{oneeq:60}
dX_{t}(\theta)=b(X_{t}(\theta),t)dt+\sigma(X_{t}(\theta),t)dB_{t},    
\end{equation}
where, $B_{t}\coloneqq W_{t}+\theta t$. Therefore, we have the following discretization:
\begin{eqnarray*}
X^{\theta}_{n+1} &=& X^{\theta}_{n}+b(X^{\theta}_{n},t_{n})h+\sigma(X^{\theta}_{n},t_{n})\Delta B_{n}+ \frac{1}{2}\sigma'(X^{\theta}_{n},t_{n}) \sigma(X^{\theta}_{n},t_{n})\left((\Delta B_{n})^{2}-h\right),\\ 
&=&X^{\theta}_{n}+\left(b(X^{\theta}_{n},t_{n})+\theta \sigma(X^{\theta}_{n},t_{n})\right)h+  \sigma(X^{\theta}_{n},t_{n})\Delta W_{n},\\
&+&\frac{1}{2}\sigma'(X^{\theta}_{n},t_{n}) \sigma(X^{\theta}_{n},t_{n}) \left((\Delta W_{n} +\theta h)^{2}-h\right).
\end{eqnarray*}
As for the Euler scheme, we have,
\begin{equation}
\label{oneeq:62}
X^{\theta}_{n+1}=X^{\theta}_{n}+\left(b(X^{\theta}_{n},t_{n})+\theta \sigma(X^{\theta}_{n},t_{n})\right)h+  \sigma(X^{\theta}_{n},t_{n})\Delta W_{n}.
\end{equation}
We use the above discretization schemes, in order to simulate the SDE in the probability space $\left(\Omega,\{\mathcal{F}_{t}^{\theta}\}_{t \geq 0},\mathbb{P}_{\theta}\right)$. Further, for the purpose of our numerical experiments, we use geometric Brownian motion $(X_{t})_{t \in [0,T]}$ to determine the value of the asset at time $t \in [0,T]$. As the process is the solution to the SDE,
\[dX_{t}=X_{t}\left(r dt+\sigma dW_{t}\right),~X_{0}=x_{0}>0,\] 
therefore, we can use the Euler or the Milstein scheme described above, in order to perform the Monte-Carlo simulations. In the above process, $r$ denotes the risk-less interest rate, $\sigma$ denotes the volatility and $W=(W_{t})_{t \in [0,T]}$ is a standard Brownian motion defined on the probability space $\left(\Omega,\{\mathcal{F}_{t}\}_{t \geq 0},\mathbb{P}\right)$.

For the purpose of estimating optimal $\theta_{l}^{*}$ for each level of resolutions, we use the algorithm described in Section \ref{stochastic_algo}. In all the numerical experiments carried out below, we consider $\displaystyle{\gamma_{n}= \frac{1}{(n+1)}}$ and $\Theta \coloneqq [0,1] \subset \mathbb{R}$. For the practical purpose, we stop the stochastic approximation procedure after finite iterations. Further we use Rupert and Poliak method for stabilizing the convergence of the described algorithm, \textit{i.e.,} instead of using $\theta_{l}^{k}$, we use $\displaystyle{\widetilde{\theta_{l}^{k}}= \frac{1}{k+1}\sum\limits_{i=0}^{k}\theta_{l}^{i}}$, on level $l$. In order to compare the results generated through ML2R and AISML2R, we define the following improvement factor ($\text{if}$) as follows,
\begin{equation}
\text{if}_{k}=\frac{\text{variance}_{\text{ml2r}} \times \text{time}_{\text{ml2r}}}{\text{variance}_{\text{aisml2r}} \times \text{time}_{\text{aisml2r}}},~k=3,\dots,9.
\end{equation}

\subsection{European option}
The payoff function in the case of European call option is described as below,
\begin{equation}
\label{oneeq:63}
P(X_{T})=e^{-rT}(X_{T}-K)_{+}.
\end{equation}
For the purpose of the practical implementation, we have considered $X_{0}=100$, $r=0.06$, $\sigma=0.4$, $T=1$ and $K=80$ \cite{lemaire2017multilevel}. With these parameters, the price obtained through closed form solution is, $\mathbf{J}_{0}= 29.4987$. Also, we have chosen $M=8$ for the Milstein scheme and $M=6$ for the Euler scheme. Further, under the Milstein discretization, we stop the stochastic algorithm after $1000$ iterations, whereas, we under go only $500$ iteration under the Euler scheme.

Table \ref{tab:table_3} and Table \ref{tab:table_4} demonstrate potency of the AISML2R over ML2R, under Milstein and Euler discretization, respectively. Figure \ref{fig:fig1} and Figure \ref{fig:fig2} graphically represents the performance of both the estimators under Milstein and Euler schemes, respectively. It is quite evident from the tabulated results, as well as the graphical representation, that the effectiveness of the AISML2R over ML2R increases with the requirement of greater accuracy. The value of $\text{if}_k$ goes from $\text{if}_{3}=0.829$, for the worst case to $\text{if}_{9}=6.78$ in the best case, while simulating under the Milstein scheme. On the other hand, under the Euler scheme, the value of $\text{if}_k$ goes from $\text{if}_{3}=0.508$ in the best case to $\text{if}_{8}=2.09$ in the best case. These values suggest that AISML2R can achieve desired root-mean-squared error much faster as compared to ML2R.   

\begin{figure}[!h]
\centering
\begin{subfigure}{0.49\textwidth}
\includegraphics[width=9cm]{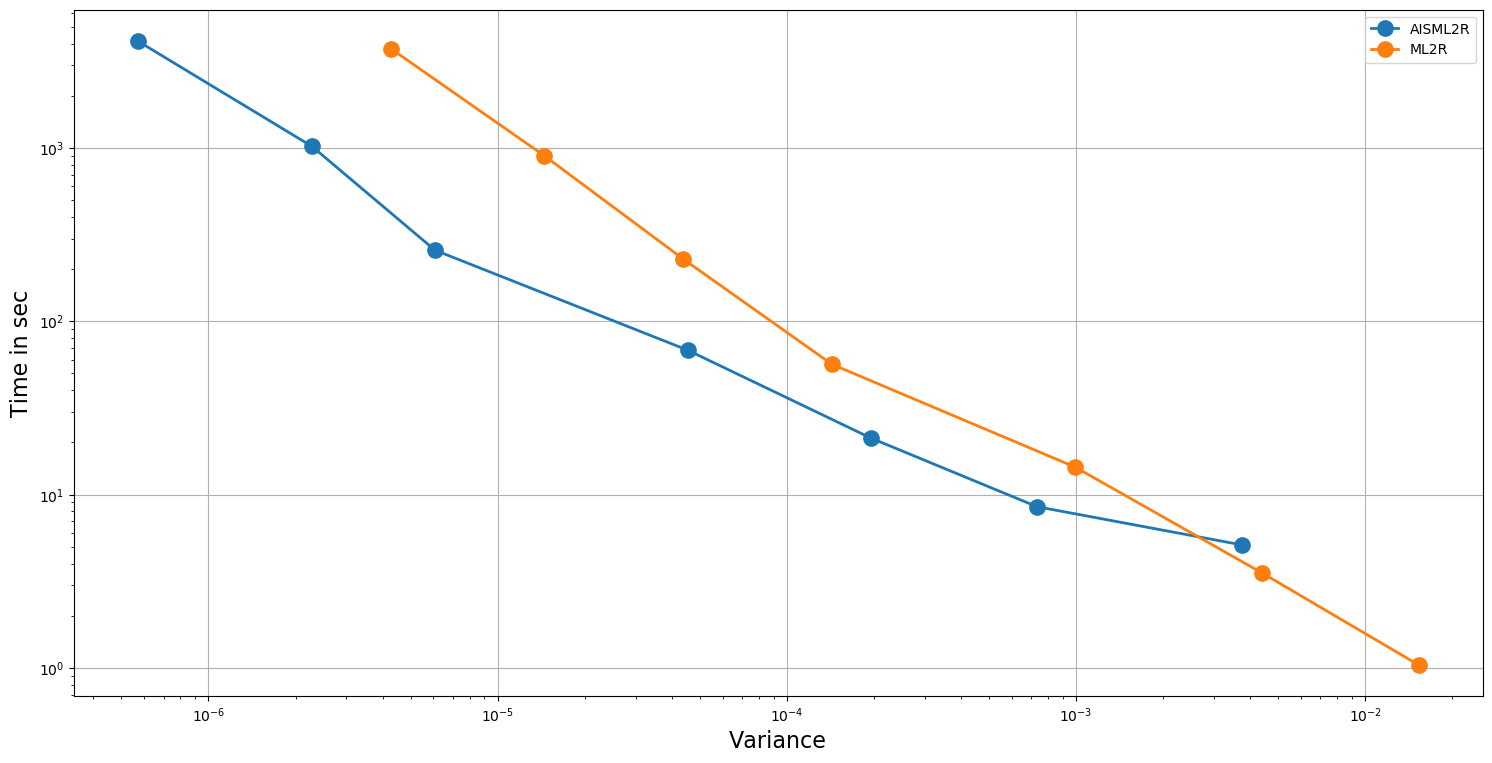}
\caption{Time (y-axis, log scale) as function of variance (x-axis, log scale)}
\end{subfigure}
\hfill
\begin{subfigure}{0.49\textwidth}
\includegraphics[width=9cm]{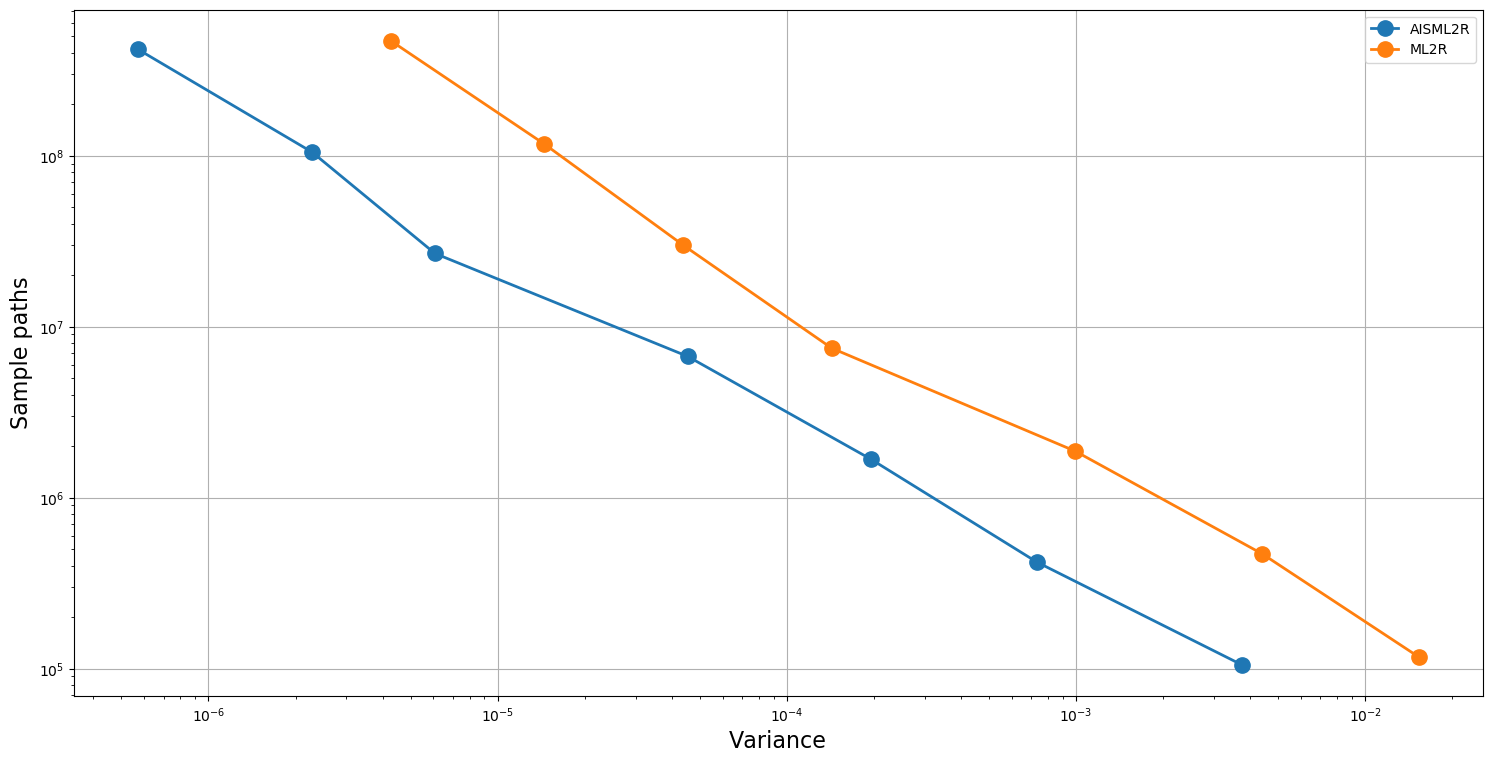}
\caption{Sample Paths (y-axis, log scale) as function of variance (x-axis, log scale)}
\end{subfigure}
\caption{European Call option using Milstein Scheme}
\label{fig:fig1}
\end{figure}

\begin{table}[!h]
\centering
\begin{tabular}{|cccccc|}
\hline
\multicolumn{1}{|c|}{}           & \multicolumn{5}{c|}{\textbf{AISML2R}}   \\ \hline
\multicolumn{1}{|c|}{\textbf{k}} & \multicolumn{1}{c|}{\textbf{N}} & \multicolumn{1}{c|}{\textbf{variance}} & \multicolumn{1}{c|}{\textbf{bias}} & \multicolumn{1}{c|}{\textbf{rmse}} & \textbf{time} \\ \hline
\multicolumn{1}{|c|}{\textbf{3}} & \multicolumn{1}{c|}{1.05E+05}   & \multicolumn{1}{c|}{3.75E-03}          & \multicolumn{1}{c|}{4.93E-02}      & \multicolumn{1}{c|}{6.18E-02}      & 5.12E+00      \\ \hline
\multicolumn{1}{|c|}{\textbf{4}} & \multicolumn{1}{c|}{4.19E+05}   & \multicolumn{1}{c|}{7.33E-04}          & \multicolumn{1}{c|}{2.21E-02}      & \multicolumn{1}{c|}{2.72E-02}      & 8.51E+00      \\ \hline
\multicolumn{1}{|c|}{\textbf{5}} & \multicolumn{1}{c|}{1.68E+06}   & \multicolumn{1}{c|}{1.95E-04}          & \multicolumn{1}{c|}{1.17E-02}      & \multicolumn{1}{c|}{1.43E-02}      & 2.12E+01      \\ \hline
\multicolumn{1}{|c|}{\textbf{6}} & \multicolumn{1}{c|}{6.70E+06}   & \multicolumn{1}{c|}{4.54E-05}          & \multicolumn{1}{c|}{5.68E-03}      & \multicolumn{1}{c|}{7.01E-03}      & 6.81E+01      \\ \hline
\multicolumn{1}{|c|}{\textbf{7}} & \multicolumn{1}{c|}{2.68E+07}   & \multicolumn{1}{c|}{6.07E-06}          & \multicolumn{1}{c|}{2.14E-03}      & \multicolumn{1}{c|}{2.48E-03}      & 2.57E+02      \\ \hline
\multicolumn{1}{|c|}{\textbf{8}} & \multicolumn{1}{c|}{1.05E+08}   & \multicolumn{1}{c|}{2.28E-06}          & \multicolumn{1}{c|}{1.14E-03}      & \multicolumn{1}{c|}{1.51E-03}      & 1.02E+03      \\ \hline
\multicolumn{1}{|c|}{\textbf{9}} & \multicolumn{1}{c|}{4.20E+08}   & \multicolumn{1}{c|}{5.69E-07}          & \multicolumn{1}{c|}{5.99E-04}      & \multicolumn{1}{c|}{7.55E-04}      & 4.14E+03      \\ \hline
\multicolumn{6}{|c|}{\textbf{ML2R}}        \\ \hline
\multicolumn{1}{|c|}{\textbf{k}}  & \multicolumn{1}{c|}{\textbf{N}} & \multicolumn{1}{c|}{\textbf{variance}} & \multicolumn{1}{c|}{\textbf{bias}} & \multicolumn{1}{c|}{\textbf{rmse}} & \textbf{time} \\ \hline
\multicolumn{1}{|c|}{\textbf{3}} & \multicolumn{1}{c|}{1.17E+05}   & \multicolumn{1}{c|}{1.53E-02}          & \multicolumn{1}{c|}{1.07E-01}      & \multicolumn{1}{c|}{1.24E-01}      & 1.04E+00      \\ \hline
\multicolumn{1}{|c|}{\textbf{4}} & \multicolumn{1}{c|}{4.68E+05}   & \multicolumn{1}{c|}{4.41E-03}          & \multicolumn{1}{c|}{5.47E-02}      & \multicolumn{1}{c|}{6.66E-02}      & 3.52E+00      \\ \hline
\multicolumn{1}{|c|}{\textbf{5}} & \multicolumn{1}{c|}{1.87E+06}   & \multicolumn{1}{c|}{9.93E-04}          & \multicolumn{1}{c|}{2.60E-02}      & \multicolumn{1}{c|}{3.17E-02}      & 1.44E+01      \\ \hline
\multicolumn{1}{|c|}{\textbf{6}} & \multicolumn{1}{c|}{7.48E+06}   & \multicolumn{1}{c|}{1.43E-04}          & \multicolumn{1}{c|}{1.05E-02}      & \multicolumn{1}{c|}{1.26E-02}      & 5.65E+01      \\ \hline
\multicolumn{1}{|c|}{\textbf{7}} & \multicolumn{1}{c|}{2.99E+07}   & \multicolumn{1}{c|}{4.39E-05}          & \multicolumn{1}{c|}{5.29E-03}      & \multicolumn{1}{c|}{6.62E-03}      & 2.28E+02      \\ \hline
\multicolumn{1}{|c|}{\textbf{8}} & \multicolumn{1}{c|}{1.17E+08}   & \multicolumn{1}{c|}{1.45E-05}          & \multicolumn{1}{c|}{2.96E-03}      & \multicolumn{1}{c|}{3.82E-03}      & 9.01E+02      \\ \hline
\multicolumn{1}{|c|}{\textbf{9}} & \multicolumn{1}{c|}{4.69E+08}   & \multicolumn{1}{c|}{4.28E-06}          & \multicolumn{1}{c|}{1.61E-03}      & \multicolumn{1}{c|}{2.14E-03}      & 3.73E+03      \\ \hline
\end{tabular}
\caption{Pricing European Option using Milstein Scheme}
\label{tab:table_3}
\end{table}

\begin{figure}[!h]
\centering
\begin{subfigure}{0.49\textwidth}
\includegraphics[width=9cm]{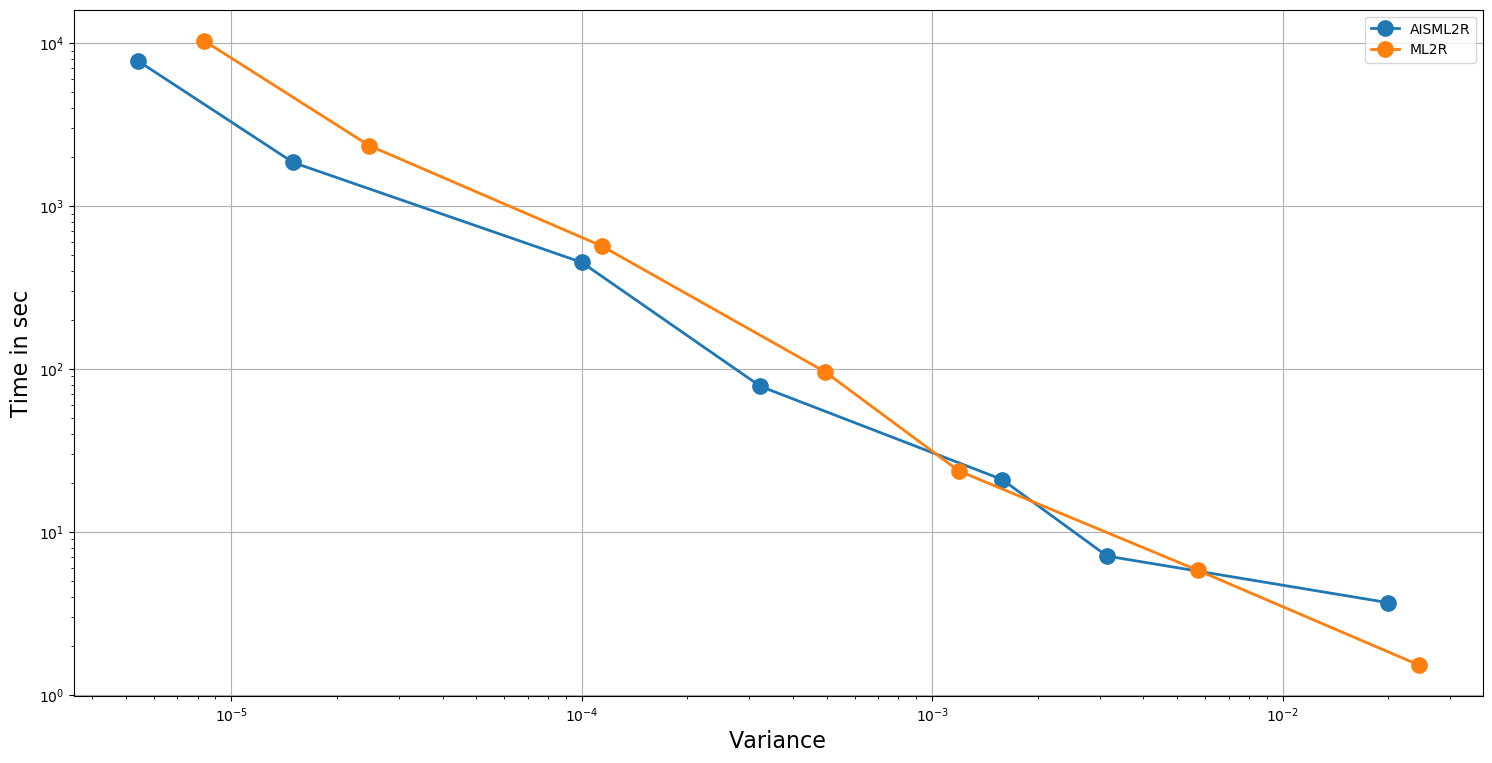}
\caption{Time (y-axis, log scale) as function of variance (x-axis, log scale)}
\end{subfigure}
\hfill
\begin{subfigure}{0.49\textwidth}
\includegraphics[width=9cm]{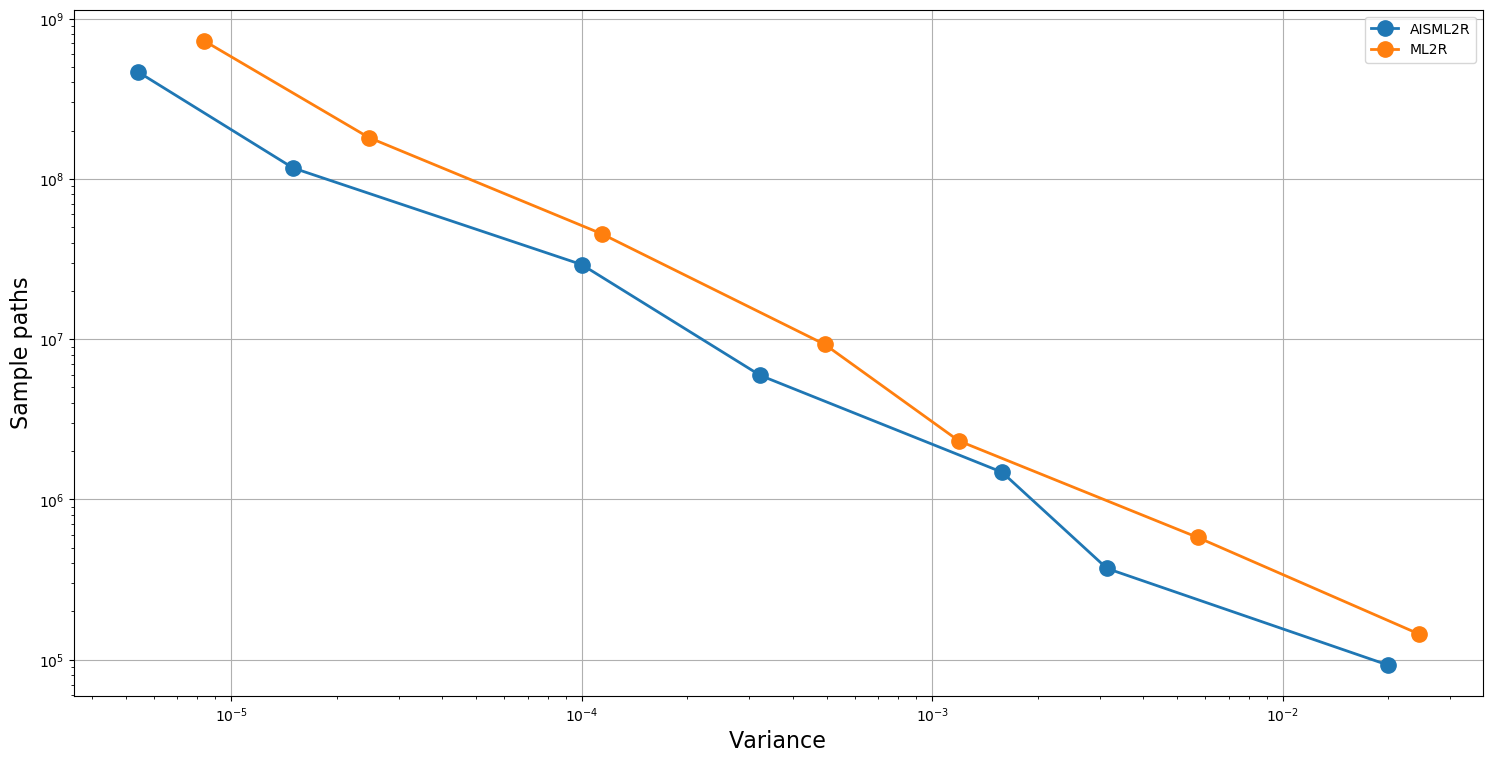}
\caption{Sample Paths (y-axis, log scale) as function of variance (x-axis, log scale)}
\end{subfigure}
\caption{European Call option using Euler Scheme}
\label{fig:fig2}
\end{figure}

\begin{table}[!h]
\centering
\begin{tabular}{|cccccc|}
\hline
\multicolumn{1}{|c|}{}           & \multicolumn{5}{c|}{\textbf{AISML2R}}      \\ \hline
\multicolumn{1}{|l|}{\textbf{k}} & \multicolumn{1}{l|}{\textbf{N}} & \multicolumn{1}{l|}{\textbf{variance}} & \multicolumn{1}{l|}{\textbf{bias}} & \multicolumn{1}{l|}{\textbf{rmse}} & \multicolumn{1}{l|}{\textbf{time}} \\ \hline
\multicolumn{1}{|r|}{\textbf{3}} & \multicolumn{1}{r|}{9.28E+04}   & \multicolumn{1}{r|}{1.99E-02}          & \multicolumn{1}{r|}{4.39E-02}      & \multicolumn{1}{r|}{1.48E-01}      & 3.69E+00                           \\ \hline
\multicolumn{1}{|r|}{\textbf{4}} & \multicolumn{1}{r|}{3.71E+05}   & \multicolumn{1}{r|}{3.15E-03}          & \multicolumn{1}{r|}{1.25E-02}      & \multicolumn{1}{r|}{5.75E-02}      & 7.12E+00                           \\ \hline
\multicolumn{1}{|r|}{\textbf{5}} & \multicolumn{1}{r|}{1.48E+06}   & \multicolumn{1}{r|}{1.58E-03}          & \multicolumn{1}{r|}{2.39E-03}      & \multicolumn{1}{r|}{3.98E-02}      & 2.10E+01                           \\ \hline
\multicolumn{1}{|r|}{\textbf{6}} & \multicolumn{1}{r|}{5.94E+06}   & \multicolumn{1}{r|}{3.23E-04}          & \multicolumn{1}{r|}{6.25E-03}      & \multicolumn{1}{r|}{1.90E-02}      & 7.83E+01                           \\ \hline
\multicolumn{1}{|r|}{\textbf{7}} & \multicolumn{1}{r|}{2.92E+07}   & \multicolumn{1}{r|}{9.99E-05}          & \multicolumn{1}{r|}{4.76E-04}      & \multicolumn{1}{r|}{1.00E-02}      & 4.51E+02                           \\ \hline
\multicolumn{1}{|r|}{\textbf{8}} & \multicolumn{1}{r|}{1.17E+08}   & \multicolumn{1}{r|}{1.50E-05}          & \multicolumn{1}{r|}{6.49E-05}      & \multicolumn{1}{r|}{3.87E-03}      & 1.85E+03                           \\ \hline
\multicolumn{1}{|r|}{\textbf{9}} & \multicolumn{1}{r|}{4.67E+08}   & \multicolumn{1}{r|}{5.40E-06}          & \multicolumn{1}{r|}{4.10E-04}      & \multicolumn{1}{r|}{2.36E-03}      & 7.79E+03                           \\ \hline
\multicolumn{6}{|c|}{\textbf{ML2R}}              \\ \hline
\multicolumn{1}{|c|}{\textbf{}}  & \multicolumn{1}{c|}{\textbf{N}} & \multicolumn{1}{c|}{\textbf{variance}} & \multicolumn{1}{c|}{\textbf{bias}} & \multicolumn{1}{c|}{\textbf{rmse}} & \multicolumn{1}{c|}{\textbf{time}} \\ \hline
\multicolumn{1}{|c|}{\textbf{3}} & \multicolumn{1}{r|}{1.45E+05}   & \multicolumn{1}{r|}{2.44E-02}          & \multicolumn{1}{r|}{5.25E-03}      & \multicolumn{1}{r|}{1.56E-01}      & 1.53E+00                           \\ \hline
\multicolumn{1}{|c|}{\textbf{4}} & \multicolumn{1}{r|}{5.80E+05}   & \multicolumn{1}{r|}{5.71E-03}          & \multicolumn{1}{r|}{5.18E-03}      & \multicolumn{1}{r|}{7.57E-02}      & 5.83E+00                           \\ \hline
\multicolumn{1}{|c|}{\textbf{5}} & \multicolumn{1}{r|}{2.32E+06}   & \multicolumn{1}{r|}{1.19E-03}          & \multicolumn{1}{r|}{8.79E-03}      & \multicolumn{1}{r|}{3.55E-02}      & 2.37E+01                           \\ \hline
\multicolumn{1}{|c|}{\textbf{6}} & \multicolumn{1}{r|}{9.28E+06}   & \multicolumn{1}{r|}{4.93E-04}          & \multicolumn{1}{r|}{2.27E-03}      & \multicolumn{1}{r|}{2.23E-02}      & 9.64E+01                           \\ \hline
\multicolumn{1}{|c|}{\textbf{7}} & \multicolumn{1}{r|}{4.54E+07}   & \multicolumn{1}{r|}{1.14E-04}          & \multicolumn{1}{r|}{1.52E-03}      & \multicolumn{1}{r|}{1.08E-02}      & 5.68E+02                           \\ \hline
\multicolumn{1}{|c|}{\textbf{8}} & \multicolumn{1}{r|}{1.81E+08}   & \multicolumn{1}{r|}{2.47E-05}          & \multicolumn{1}{r|}{8.38E-05}      & \multicolumn{1}{r|}{4.97E-03}      & 2.35E+03                           \\ \hline
\multicolumn{1}{|c|}{\textbf{9}} & \multicolumn{1}{r|}{7.26E+08}   & \multicolumn{1}{r|}{8.37E-06}          & \multicolumn{1}{r|}{6.35E-04}      & \multicolumn{1}{r|}{2.96E-03}      & 1.03E+04                           \\ \hline
\end{tabular}
\caption{Pricing European Option using Euler Scheme}
\label{tab:table_4}
\end{table}

\subsection{Lookback Option}

In this section we consider a partial lookback call option, whose payoff is defined as,
\begin{equation}
P(X_{T})=e^{-rT}(X_{T}-\zeta \min_{t\in[0,T]}X(t))_{+},~\text{where}~\zeta \geq 1.
\end{equation}
The parameters for the purpose of the practical implementation, taken from \cite{lemaire2017multilevel}, are $X_{0}=100$, $r=0.15$, $\sigma=0.1$ and $T=1$. Also, the value of $\zeta$ is set as $\zeta=1.1$. Further, the refinement factor $M$ used to perform the simulation at various level of dicretization, is set as $M=8$, for either case. With these parameters, the price obtained through closed form solution is $\mathbf{J}_{0}=8.89343$. The simulation is carried out using both Milstein and Euler scheme, as described above. In order to exploit the full essence of Milstein discretization, we use the formula described below in order to simulate the value of $\min_{t \in [0,T]}X(t)$ on various refinement level, \cite{giles2008improved}.
\begin{equation}
X^{n_l}_{n,\min}=\frac{1}{2}\left(X^{n_{l}}_{n}+X^{n_{l}}_{n+1}-\sqrt{\left(X^{n_{l}}_{n+1}-X^{n_{l}}_{n}\right)^{2}-2 \left(\sigma X^{n_{l}}_{n}\right)^{2}h_{l}\log U_{n}}\right),~ U_{n}\sim \text{Uinf}(0,1).
\end{equation} 
Further, we set $X^{n_{l-1}}_{n,\min}=X^{n_l}_{n,\min}$. In this case we stop the stochastic algorithm after $200$ iterations for either case.

The results summarized in Table \ref{tab:table_5} and Table \ref{tab:table_6} shows the efficacy of AISML2R over ML2R in the case of both Euler and Milstein discretization scheme, respectively. Figure \ref{fig:fig3} and Figure \ref{fig:fig4} graphically represents the performance of both estimators under Milstein and Euler schemes, respectively. Similar to the case of European Option, we can easily observe that the efficiency of AISML2R over ML2R increases as more accuracy is required. The value of $\text{if}$ goes from $\text{if}_{3}=0.365$ in the worst case to $\text{if}_{8}=3.81$ in the best case, while simulating under Milstein scheme. Under the Euler scheme, the value goes from $\text{if}_{3}=0.333$ in the worst case to $\text{if}_{6}=2.58$ in the best case.   
\begin{table}[!h]
\centering
\begin{tabular}{|cccccc|}
\hline
\multicolumn{1}{|c|}{}           & \multicolumn{5}{c|}{\textbf{AISML2R}}   \\ \hline
\multicolumn{1}{|l|}{\textbf{k}} & \multicolumn{1}{l|}{\textbf{N}} & \multicolumn{1}{l|}{\textbf{variance}} & \multicolumn{1}{l|}{\textbf{bias}} & \multicolumn{1}{l|}{\textbf{rmse}} & \multicolumn{1}{l|}{\textbf{time}} \\ \hline
\multicolumn{1}{|r|}{\textbf{3}} & \multicolumn{1}{r|}{2.86E+03}   & \multicolumn{1}{r|}{1.44E-02}          & \multicolumn{1}{r|}{2.41E-02}      & \multicolumn{1}{r|}{1.22E-01}      & 7.99E-01                           \\ \hline
\multicolumn{1}{|r|}{\textbf{4}} & \multicolumn{1}{r|}{1.14E+04}   & \multicolumn{1}{r|}{3.51E-03}          & \multicolumn{1}{r|}{2.13E-03}      & \multicolumn{1}{r|}{5.93E-02}      & 1.03E+00                           \\ \hline
\multicolumn{1}{|r|}{\textbf{5}} & \multicolumn{1}{r|}{4.58E+04}   & \multicolumn{1}{r|}{7.14E-04}          & \multicolumn{1}{r|}{8.31E-03}      & \multicolumn{1}{r|}{2.80E-02}      & 1.67E+00                           \\ \hline
\multicolumn{1}{|r|}{\textbf{6}} & \multicolumn{1}{r|}{1.83E+05}   & \multicolumn{1}{r|}{1.85E-04}          & \multicolumn{1}{r|}{9.92E-03}      & \multicolumn{1}{r|}{1.68E-02}      & 3.78E+00                           \\ \hline
\multicolumn{1}{|r|}{\textbf{7}} & \multicolumn{1}{r|}{7.33E+05}   & \multicolumn{1}{r|}{6.73E-05}          & \multicolumn{1}{r|}{9.88E-03}      & \multicolumn{1}{r|}{1.28E-02}      & 1.32E+01                           \\ \hline
\multicolumn{1}{|r|}{\textbf{8}} & \multicolumn{1}{r|}{3.39E+06}   & \multicolumn{1}{r|}{1.45E-05}          & \multicolumn{1}{r|}{8.88E-03}      & \multicolumn{1}{r|}{9.66E-03}      & 7.14E+01                           \\ \hline
\multicolumn{1}{|r|}{\textbf{9}} & \multicolumn{1}{r|}{1.36E+07}   & \multicolumn{1}{r|}{5.06E-06}          & \multicolumn{1}{r|}{8.81E-03}      & \multicolumn{1}{r|}{9.09E-03}      & 2.77E+02                           \\ \hline
\multicolumn{6}{|c|}{\textbf{ML2R}}                  \\ \hline
\multicolumn{1}{|c|}{\textbf{k}} & \multicolumn{1}{c|}{\textbf{N}} & \multicolumn{1}{c|}{\textbf{variance}} & \multicolumn{1}{c|}{\textbf{bias}} & \multicolumn{1}{c|}{\textbf{rmse}} & \multicolumn{1}{c|}{\textbf{time}} \\ \hline
\multicolumn{1}{|l|}{\textbf{3}} & \multicolumn{1}{r|}{6.43E+03}   & \multicolumn{1}{r|}{1.29E-02}          & \multicolumn{1}{r|}{1.45E-02}      & \multicolumn{1}{r|}{1.15E-01}      & 2.97E-01                           \\ \hline
\multicolumn{1}{|l|}{\textbf{5}} & \multicolumn{1}{r|}{2.57E+04}   & \multicolumn{1}{r|}{4.37E-03}          & \multicolumn{1}{r|}{1.17E-03}      & \multicolumn{1}{r|}{6.61E-02}      & 6.39E-01                           \\ \hline
\multicolumn{1}{|c|}{\textbf{5}} & \multicolumn{1}{r|}{1.03E+05}   & \multicolumn{1}{r|}{8.24E-04}          & \multicolumn{1}{r|}{1.19E-02}      & \multicolumn{1}{r|}{3.11E-02}      & 1.91E+00                           \\ \hline
\multicolumn{1}{|c|}{\textbf{6}} & \multicolumn{1}{r|}{4.12E+05}   & \multicolumn{1}{r|}{2.43E-04}          & \multicolumn{1}{r|}{8.33E-03}      & \multicolumn{1}{r|}{1.77E-02}      & 7.42E+00                           \\ \hline
\multicolumn{1}{|c|}{\textbf{7}} & \multicolumn{1}{r|}{1.65E+06}   & \multicolumn{1}{r|}{6.59E-05}          & \multicolumn{1}{r|}{9.35E-03}      & \multicolumn{1}{r|}{1.24E-02}      & 3.02E+01                           \\ \hline
\multicolumn{1}{|c|}{\textbf{8}} & \multicolumn{1}{r|}{7.93E+06}   & \multicolumn{1}{r|}{1.39E-05}          & \multicolumn{1}{r|}{9.05E-03}      & \multicolumn{1}{r|}{9.79E-03}      & 1.76E+02                           \\ \hline
\multicolumn{1}{|c|}{\textbf{9}} & \multicolumn{1}{r|}{3.17E+07}   & \multicolumn{1}{r|}{2.62E-06}          & \multicolumn{1}{r|}{9.08E-03}      & \multicolumn{1}{r|}{9.22E-03}      & 7.10E+02                           \\ \hline
\end{tabular}
\caption{Pricing Lookback Option using Euler Scheme}
\label{tab:table_5}
\end{table}

\begin{figure}[!h]
\centering
\begin{subfigure}{0.49\textwidth}
\includegraphics[width=9cm]{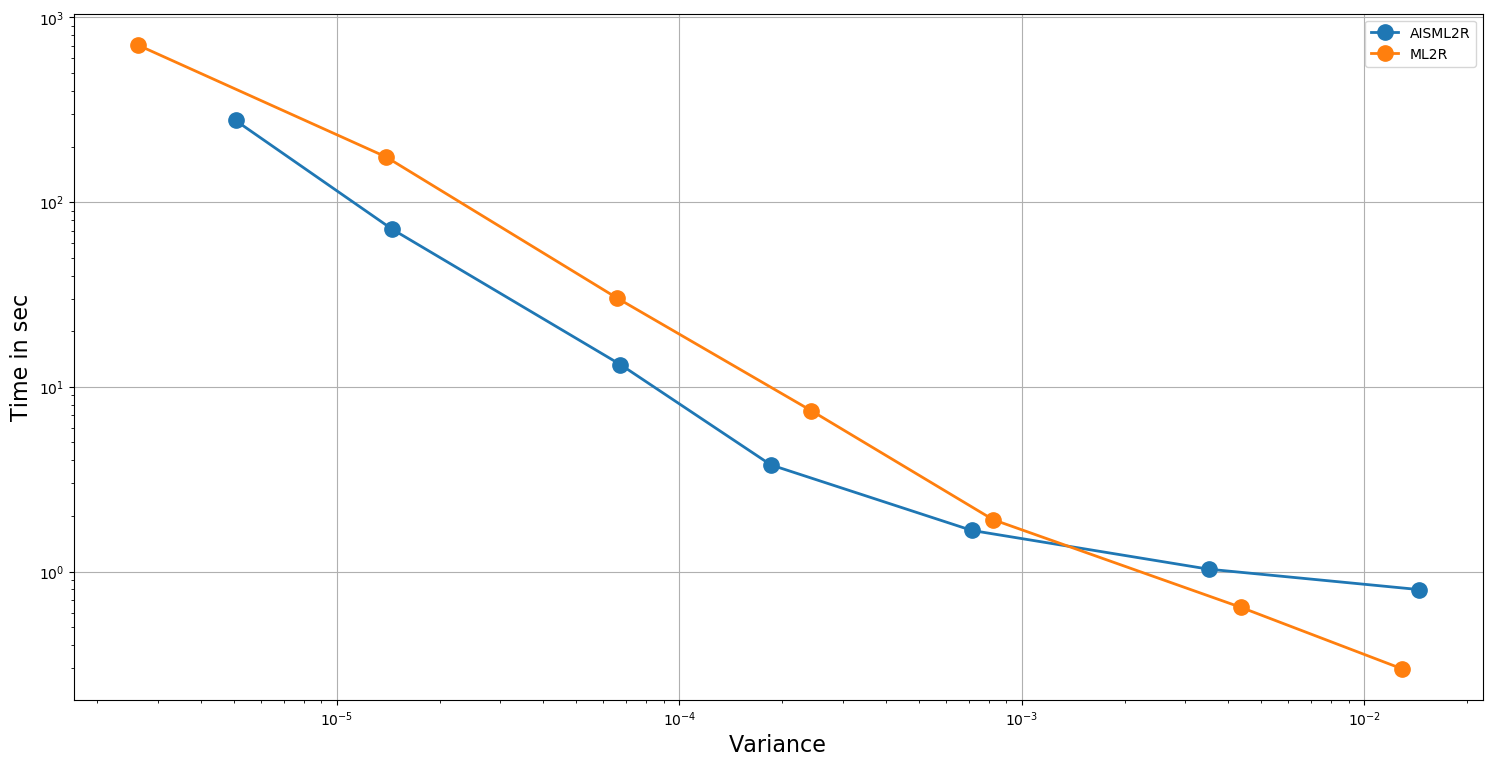}
\caption{Time (y-axis, log scale) as function of variance (x-axis, log scale)}
\end{subfigure}
\hfill
\begin{subfigure}{0.49\textwidth}
\includegraphics[width=9cm]{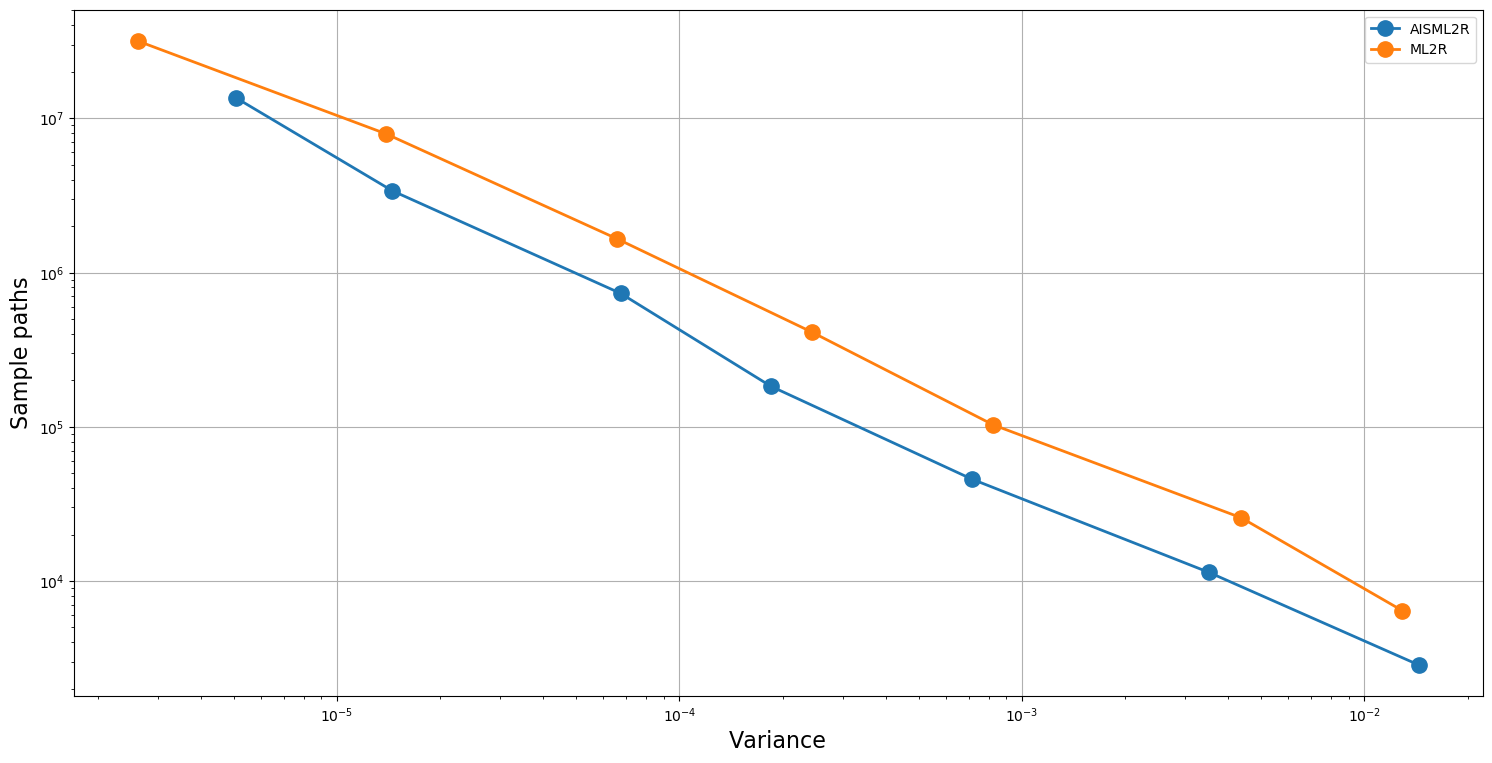}
\caption{Sample Paths (y-axis, log scale) as function of variance (x-axis, log scale)}
\end{subfigure}
\caption{Pricing Lookback Call option using Euler Scheme}
\label{fig:fig3}
\end{figure}

\begin{table}[]
\centering
\begin{tabular}{|cccccc|}
\hline
\multicolumn{1}{|c|}{}           & \multicolumn{5}{c|}{\textbf{AISML2R}}    \\ \hline
\multicolumn{1}{|l|}{\textbf{k}} & \multicolumn{1}{l|}{\textbf{N}} & \multicolumn{1}{l|}{\textbf{variance}} & \multicolumn{1}{l|}{\textbf{bias}} & \multicolumn{1}{l|}{\textbf{rmse}} & \multicolumn{1}{l|}{\textbf{time}} \\ \hline
\multicolumn{1}{|r|}{\textbf{3}} & \multicolumn{1}{r|}{1.59E+03}   & \multicolumn{1}{r|}{1.29E-02}          & \multicolumn{1}{r|}{7.52E-02}      & \multicolumn{1}{r|}{1.36E-01}      & 7.71E-01                           \\ \hline
\multicolumn{1}{|r|}{\textbf{4}} & \multicolumn{1}{r|}{6.38E+03}   & \multicolumn{1}{r|}{4.31E-03}          & \multicolumn{1}{r|}{2.92E-02}      & \multicolumn{1}{r|}{7.19E-02}      & 9.21E-01                           \\ \hline
\multicolumn{1}{|r|}{\textbf{5}} & \multicolumn{1}{r|}{2.55E+04}   & \multicolumn{1}{r|}{1.01E-03}          & \multicolumn{1}{r|}{2.99E-02}      & \multicolumn{1}{r|}{4.36E-02}      & 1.30E+00                           \\ \hline
\multicolumn{1}{|r|}{\textbf{6}} & \multicolumn{1}{r|}{1.02E+05}   & \multicolumn{1}{r|}{2.27E-04}          & \multicolumn{1}{r|}{2.64E-02}      & \multicolumn{1}{r|}{3.04E-02}      & 2.52E+00                           \\ \hline
\multicolumn{1}{|r|}{\textbf{7}} & \multicolumn{1}{r|}{4.08E+05}   & \multicolumn{1}{r|}{7.20E-05}          & \multicolumn{1}{r|}{2.43E-02}      & \multicolumn{1}{r|}{2.57E-02}      & 7.40E+00                           \\ \hline
\multicolumn{1}{|r|}{\textbf{8}} & \multicolumn{1}{r|}{1.61E+06}   & \multicolumn{1}{r|}{1.14E-05}          & \multicolumn{1}{r|}{2.36E-02}      & \multicolumn{1}{r|}{2.38E-02}      & 3.27E+01                           \\ \hline
\multicolumn{1}{|r|}{\textbf{9}} & \multicolumn{1}{r|}{6.44E+06}   & \multicolumn{1}{r|}{4.40E-06}          & \multicolumn{1}{r|}{2.27E-02}      & \multicolumn{1}{r|}{2.28E-02}      & 1.17E+02                           \\ \hline
\multicolumn{6}{|c|}{\textbf{ML2R}}              \\ \hline
\multicolumn{1}{|c|}{\textbf{k}} & \multicolumn{1}{c|}{\textbf{N}} & \multicolumn{1}{c|}{\textbf{variance}} & \multicolumn{1}{c|}{\textbf{bias}} & \multicolumn{1}{c|}{\textbf{rmse}} & \multicolumn{1}{c|}{\textbf{time}} \\ \hline
\multicolumn{1}{|l|}{\textbf{3}} & \multicolumn{1}{r|}{5.16E+03}   & \multicolumn{1}{r|}{1.24E-02}          & \multicolumn{1}{r|}{2.65E-02}      & \multicolumn{1}{r|}{1.14E-01}      & 2.93E-01                           \\ \hline
\multicolumn{1}{|l|}{\textbf{5}} & \multicolumn{1}{r|}{2.06E+04}   & \multicolumn{1}{r|}{3.00E-03}          & \multicolumn{1}{r|}{2.76E-02}      & \multicolumn{1}{r|}{6.13E-02}      & 5.02E-01                           \\ \hline
\multicolumn{1}{|c|}{\textbf{5}} & \multicolumn{1}{r|}{8.25E+04}   & \multicolumn{1}{r|}{8.71E-04}          & \multicolumn{1}{r|}{2.90E-02}      & \multicolumn{1}{r|}{4.14E-02}      & 1.33E+00                           \\ \hline
\multicolumn{1}{|c|}{\textbf{6}} & \multicolumn{1}{r|}{3.30E+05}   & \multicolumn{1}{r|}{4.09E-04}          & \multicolumn{1}{r|}{2.34E-02}      & \multicolumn{1}{r|}{3.09E-02}      & 4.93E+00                           \\ \hline
\multicolumn{1}{|c|}{\textbf{7}} & \multicolumn{1}{r|}{1.32E+06}   & \multicolumn{1}{r|}{7.61E-05}          & \multicolumn{1}{r|}{2.48E-02}      & \multicolumn{1}{r|}{2.63E-02}      & 1.99E+01                           \\ \hline
\multicolumn{1}{|c|}{\textbf{8}} & \multicolumn{1}{r|}{5.21E+06}   & \multicolumn{1}{r|}{1.71E-05}          & \multicolumn{1}{r|}{2.39E-02}      & \multicolumn{1}{r|}{2.42E-02}      & 8.31E+01                           \\ \hline
\multicolumn{1}{|c|}{\textbf{9}} & \multicolumn{1}{r|}{2.08E+07}   & \multicolumn{1}{r|}{4.02E-06}          & \multicolumn{1}{r|}{2.26E-02}      & \multicolumn{1}{r|}{2.27E-02}      & 3.30E+02                           \\ \hline
\end{tabular}
\caption{Pricing Lookback Option using Milstein Scheme}
\label{tab:table_6}
\end{table}

\begin{figure}[!h]
\centering
\begin{subfigure}{0.49\textwidth}
\includegraphics[width=9cm]{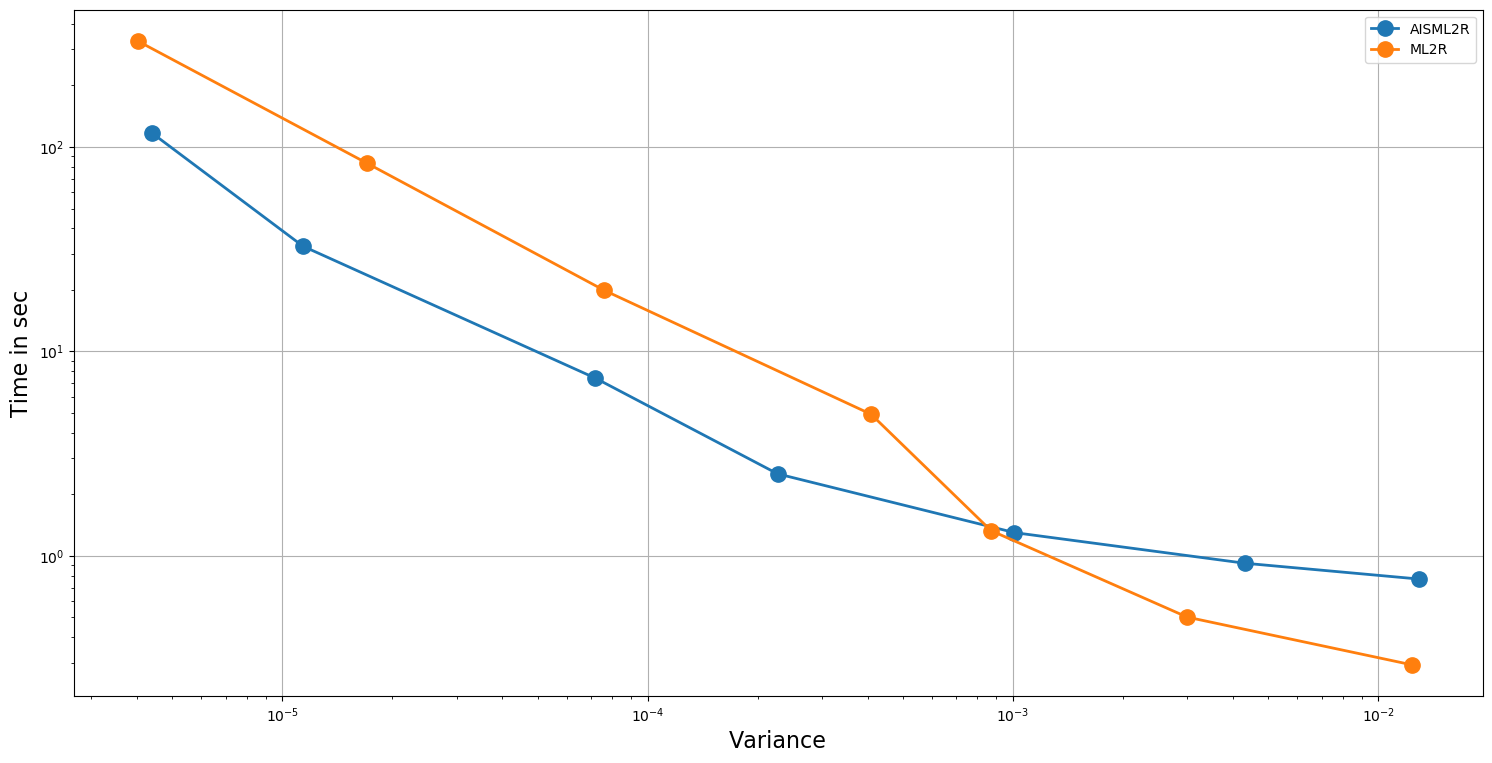}
\caption{Time (y-axis, log scale) as function of variance (x-axis, log scale)}
\end{subfigure}
\hfill
\begin{subfigure}{0.49\textwidth}
\includegraphics[width=9cm]{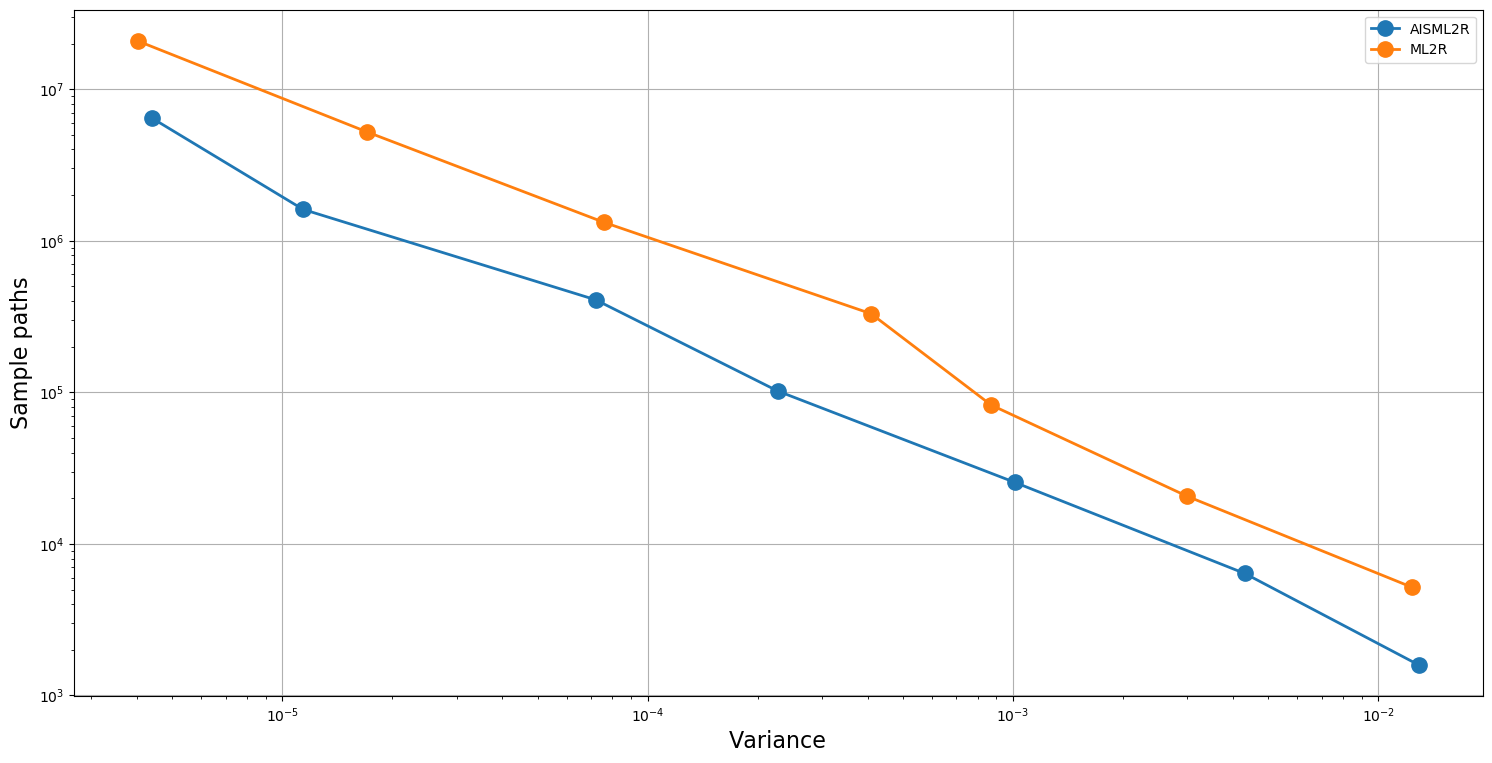}
\caption{Sample Paths (y-axis, log scale) as function of variance(x-axis, log scale)}
\end{subfigure}
\caption{Pricing Lookback Call option using Milstein Scheme}
\label{fig:fig4}
\end{figure}

\newpage

\section{Conclusion}
This paper presents a novel combination of ML2R and importance sampling algorithm developed upon the studies carried out in \cite{lemaire2017multilevel,arouna2004adaptative,alaya2015importance,alaya2016improved,kebaier2018coupling}. As evident from the numerical results presented in Section \ref{section_numerical}, the adaptive estimator outperforms AISML2R whenever a high degree of accuracy is required. One of the critical features of this, presented in the paper, is the use of a higher order scheme for discretizing the underlying SDE. Therefore, accommodating the antithetic feature developed by Giles \cite{giles2013antithetic} to simulate higher dimensional SDEs, using the Milstein scheme, in the realm of the above-developed algorithm, can substantially affect the overall computational time. Further, the completely automated nature of the algorithm gives its edge over the already existing estimators. However, it may be pointed out that a small amount of fine-tuning may be required while dealing with the Euler discretization scheme.
In this paper, besides establishing the existence and uniqueness of the optimal parameter on various levels of resolutions (Section \ref{section_existanduni}), we have also proved the Strong Law of Large Numbers in Section \ref{section_mainresult}, proving the convergence of the adaptive estimator to the standard expectation. However, proving the Central Limit theorem still requires further detailed study of our estimator.

\newpage

\bibliographystyle{ieeetr}

\bibliography{BIBLIO_01}

\begin{thebibliography}{10}

\bibitem{giles2008multilevel}
M.~B. Giles, ``Multilevel monte carlo path simulation,'' {\em Operations
  research}, vol.~56, no.~3, pp.~607--617, 2008.

\bibitem{heinrich2001multilevel}
S.~Heinrich, ``Multilevel monte carlo methods,'' in {\em International
  Conference on Large-Scale Scientific Computing}, pp.~58--67, Springer, 2001.

\bibitem{giles2008improved}
M.~Giles, ``Improved multilevel monte carlo convergence using the milstein
  scheme,'' in {\em Monte Carlo and Quasi-Monte Carlo Methods 2006},
  pp.~343--358, Springer, 2008.

\bibitem{giles2019analysis}
M.~B. Giles, K.~Debrabant, and A.~R{\"o}ssler, ``Analysis of multilevel monte
  carlo path simulation using the milstein discretisation,'' {\em Discrete \&
  Continuous Dynamical Systems-B}, vol.~24, no.~8, p.~3881, 2019.

\bibitem{giles2013multilevel}
M.~Giles and L.~Szpruch, ``Multilevel monte carlo methods for applications in
  finance,'' {\em Recent Developments in Computational Finance: Foundations,
  Algorithms and Applications}, pp.~3--47, 2013.

\bibitem{giles2013antithetic}
M.~B. Giles and L.~Szpruch, ``Antithetic multilevel monte carlo estimation for
  multidimensional sdes,'' in {\em Monte Carlo and Quasi-Monte Carlo Methods
  2012}, pp.~367--384, Springer, 2013.

\bibitem{giles2009multilevel}
M.~B. Giles and B.~J. Waterhouse, ``Multilevel quasi-monte carlo path
  simulation,'' in {\em Advanced financial modelling}, pp.~165--182, De
  Gruyter, 2009.

\bibitem{giles2015multilevel}
M.~B. Giles, ``Multilevel monte carlo methods,'' {\em Acta Numerica}, vol.~24,
  pp.~259--328, 2015.

\bibitem{rhee2015unbiased}
C.-h. Rhee and P.~W. Glynn, ``Unbiased estimation with square root convergence
  for sde models,'' {\em Operations Research}, vol.~63, no.~5, pp.~1026--1043,
  2015.

\bibitem{lemaire2017multilevel}
V.~Lemaire and G.~Pag{\`e}s, ``Multilevel richardson--romberg extrapolation,''
  {\em Bernoulli}, vol.~23, no.~4A, pp.~2643--2692, 2017.

\bibitem{kloeden1992stochastic}
P.~E. Kloeden and E.~Platen, ``Stochastic differential equations,'' in {\em
  Numerical Solution of Stochastic Differential Equations}, pp.~103--160,
  Springer, 1992.

\bibitem{giorgi2017limit}
D.~Giorgi, V.~Lemaire, and G.~Pag{\`e}s, ``Limit theorems for weighted and
  regular multilevel estimators,'' {\em Monte Carlo Methods and Applications},
  vol.~23, no.~1, pp.~43--70, 2017.

\bibitem{arouna2004adaptative}
B.~Arouna, ``Adaptative monte carlo method, a variance reduction technique,''
  2004.

\bibitem{glasserman2004monte}
P.~Glasserman, {\em Monte Carlo methods in financial engineering}, vol.~53.
\newblock Springer, 2004.

\bibitem{alaya2015importance}
M.~B. Alaya, K.~Hajji, and A.~Kebaier, ``Importance sampling and statistical
  romberg method,'' {\em Bernoulli}, vol.~21, no.~4, pp.~1947--1983, 2015.

\bibitem{chen1987convergence}
H.-F. Chen, L.~Guo, and A.-J. Gao, ``Convergence and robustness of the
  robbins-monro algorithm truncated at randomly varying bounds,'' {\em
  Stochastic Processes and their Applications}, vol.~27, pp.~217--231, 1987.

\bibitem{chen1986stochastic}
H.F.Chen and Y.~Zhu, ``Stochastic approximation procedures with randomly
  varying truncations,'' {\em Science in China, Ser. A}, 1986.

\bibitem{andrieu2005stability}
C.~Andrieu, {\'E}.~Moulines, and P.~Priouret, ``Stability of stochastic
  approximation under verifiable conditions,'' {\em SIAM Journal on control and
  optimization}, vol.~44, no.~1, pp.~283--312, 2005.

\bibitem{lelong2008almost}
J.~Lelong, ``Almost sure convergence of randomly truncated stochastic
  algorithms under verifiable conditions,'' {\em Statistics \& Probability
  Letters}, vol.~78, no.~16, pp.~2632--2636, 2008.

\bibitem{lemaire2010unconstrained}
V.~Lemaire and G.~Pag{\`e}s, ``Unconstrained recursive importance sampling,''
  {\em The Annals of Applied Probability}, vol.~20, no.~3, pp.~1029--1067,
  2010.

\bibitem{alaya2016improved}
M.~B. Alaya, K.~Hajji, and A.~Kebaier, ``Improved adaptive multilevel monte
  carlo and applications to finance,'' {\em arXiv preprint arXiv:1603.02959},
  2016.

\bibitem{kebaier2018coupling}
A.~Kebaier and J.~Lelong, ``Coupling importance sampling and multilevel monte
  carlo using sample average approximation,'' {\em Methodology and Computing in
  Applied Probability}, vol.~20, no.~2, pp.~611--641, 2018.

\bibitem{alaya2015central}
M.~B. Alaya and A.~Kebaier, ``Central limit theorem for the multilevel monte
  carlo euler method,'' {\em The Annals of Applied Probability}, vol.~25,
  no.~1, pp.~211--234, 2015.

\bibitem{laruelle2013optimal}
S.~Laruelle, C.-A. Lehalle, {\em et~al.}, ``Optimal posting price of limit
  orders: learning by trading,'' {\em Mathematics and Financial Economics},
  vol.~7, no.~3, pp.~359--403, 2013.

\bibitem{hall2014martingale}
P.~Hall and C.~C. Heyde, {\em Martingale limit theory and its application}.
\newblock Academic press, 2014.

\end{thebibliography}

\end{document}